\documentclass{cpamart}     
\received{Month 200X}       
\volume{000}  
\startingpage{1}                        
  
\authorheadline{A. Biryuk \& W. Craig}  
\titleheadline{Bounds on Kolmogorov spectra}  
  
\usepackage{latexsym}
\usepackage{graphics}
\usepackage{amsfonts,epsf,amssymb,amscd} 
\usepackage{epsfig} 
  
\usepackage{mathptmx}  
  
\newtheorem{theorem}{Theorem}[section]  
\newtheorem{lemma}[theorem]{Lemma}  
  
\newtheorem{proposition}[theorem]{Proposition}  
\theoremstyle{definition}  
\newtheorem{definition}[theorem]{Definition}  
\theoremstyle{remark}

\newcommand{\bfg}{\begin{figure}}
\newcommand{\efg}{\end{figure}}
\newcommand{\Half}{{\textstyle{\frac{1}{2}}}}
  
\begin{document}                        
  
  
\title{Bounds on Kolmogorov spectra for the \\ Navier -- Stokes equations}  
  
  
\author{Andrei Biryuk 
\affil{Credit Suisse}
%
%
\\ Walter Craig 
\affil{McMaster University}
}
  
  
  
\begin{abstract}  
 Let $u(x,t)$ be a (possibly weak) solution of the Navier - Stokes equations  
on all of ${\mathbb R}^3$, or on the torus ${\mathbb R}^3/ {\mathbb Z}^3$. 
The {\it energy spectrum} of $u(\cdot,t)$ is the spherical integral
\[
   E(\kappa,t) = \int_{|k| = \kappa} |\hat{u}(k,t)|^2 dS(k) , \qquad
   0 \leq \kappa < \infty ,  
\]
or alternatively, a suitable approximate sum. An argument 
involking scale invariance and dimensional analysis  
given by Kolmogorov~\cite{K41a,K41b} and Obukhov~\cite{Obukhov1941}  
predicts that large Reynolds number solutions of the 
Navier - Stokes equations in three dimensions should obey 
\[
  E(\kappa, t) \sim C_0\varepsilon^{2/3}\kappa^{-5/3} 
\]
over an {\it inertial range} $\kappa_1 \leq \kappa \leq \kappa_2$,
at least in an average sense. We give a global estimate on weak 
solutions in the norm 
   $\|{\mathcal F}\partial_x u(\cdot, t)\|_\infty$ 
which gives bounds on a solution's ability to satisfy the 
Kolmogorov law. A subsequent result is for rigorous upper and lower
bounds on the inertial range, and an upper bound on the time 
of validity of the Kolmogorov spectral regime.  
\end{abstract}  
  
\maketitle

  
\section{Introduction}
\label{Introduction}

An important issue in the study of solutions of the Navier -- Stokes 
equations in the large is the principle governing the distribution 
of energy in Fourier space. The theory of Kolmogorov \cite{K41a,K41c,K41b}
and Obukhov \cite{Obukhov1941} plays a central r\^ole, predicting power
law decay behavior of the Fourier space energy density for solutions
which exhibit fully developed turbulence. In outline, a basic 
prediction is that energy spectral functions $E(\kappa,t)$, or 
possibly its average over a statistical ensemble,  is expected 
to satisfy 
\begin{equation}\label{Eqn:KolmogorovObukhovLaw}
   E(\kappa,t) \simeq C_0\varepsilon^{2/3}\kappa^{-5/3}
\end{equation}
over an inertial range of wavenumbers $\kappa \in [\kappa_1,\kappa_2]$, 
where $C_0$ is a dimensionless constant, $\varepsilon$ is a parameter 
interpreted physically as the energy transfer rate per unit volume, 
and the exponents are determined by dimensional analysis 
\cite{Obukhov1941}\cite{Frisch1995}. This famous statement has been
very influential in the field, and considerable experimental and
numerical evidence has been gathered to support it. Despite its
success, there have been relatively few rigorous mathematical results
on the analysis of solutions of the Navier -- Stoke equations, with or
without bulk {\it inhomogeneous} forces, which have addressed the 
question as to whether solutions exhibit spectral behavior as
described by \eqref{Eqn:KolmogorovObukhovLaw}. Among those papers 
which do address certain aspects of these questions, we cite in 
particular two sources. Firstly, the book by 
C.~Doering \& J.~Gibbon~\cite{DoeringGibbon95} reviews the 
Kolmogorov -- Obukhov theory, and discusses the compatibility of 
spectral aspects of solutions with the $L^2$ regularity theory for 
the Navier --- Stokes equations. Secondly, S. Kuksin~\cite{Kuksin99}
proves that solutions to the nonlinear Schr\"odinger equation with 
added dissipation and with stochastic forces exhibit spectral 
behavior over an inertial range, with some positive exponent
(which is not known explicitly). This latter work serves as an
important mathematical model of generation of spectral behavior 
of solutions under stochastic forcing, despite the basic difference 
in the equations that are addressed. 

In this paper we give a new global estimate in the norm 
$\|{\mathcal F}\, \partial_x u(\cdot, t) \|_{L^\infty}$
on weak solutions of the Navier -- Stokes equations which have  
reasonably smooth initial data and which are possibly subject to 
reasonably smooth inhomogeneous forces. 
This estimate has implications on the energy spectral function
for such solutions, and in particular in the case that there is no
inhomogeneous force, we show that weak solutions of the initial value
problem have spectral energy function which for all 
$\kappa \in {\mathbb R}^+$ satisfies 
\begin{equation}\label{Eqn:UpperBoundE}
    E(\kappa,t) \leq 4\pi R_1^2 ~,
\end{equation}
and time averages which satisfy
\begin{equation}\label{Eqn:TmeAverageE}
   \frac{1}{T}\int_0^T E(\kappa,t) \, dt  
   \leq \frac{4\pi R_2^2}{\nu\kappa^2T}  ~,
\end{equation}
again for all $\kappa$, where $\nu$ is the coefficient of viscosity.
In the case that a bounded inhomogeneous force is applied to the
solution of the initial value problem, we find similarly that
\begin{equation}\label{Eqn:UpperBoundE2}
    E(\kappa,t) \leq 4\pi R_1^2(t) ~,
\end{equation}
and furthermore 
\begin{equation}\label{Eqn:TmeAverageE2}
   \frac{1}{T}\int_0^T E(\kappa,t) \, dt  
   \leq \frac{4\pi R_2^2(T)}{\nu\kappa^2 T}  ~.
\end{equation}
In the situation of forcing being given by a stationary process, 
it is to be expected that the quantity $R_2^2(T)/T$ has a limit 
$\overline{R}_2^2$ for large $T$, giving a constant upper bound 
for the time average of $E(\kappa,T)$.
Since these estimates give rigorous bounds on $E(\kappa,t)$ with 
a faster rate of decay in wavenumber $\kappa$ than 
\eqref{Eqn:KolmogorovObukhovLaw}, this result presents a conundrum. 
Either it is the case that solutions which exhibit large scale 
spectral behavior as in \eqref{Eqn:KolmogorovObukhovLaw} are
not smooth, and in particular do not arise from the initial value
problem with reasonably smooth initial data. Or else the bounds 
\eqref{Eqn:UpperBoundE}\eqref{Eqn:TmeAverageE} (and 
\eqref{Eqn:UpperBoundE2}\eqref{Eqn:TmeAverageE2} respectively, 
in the case with
inhomogeneous forces) give restrictions on the spectral behavior of
solutions, and in particular an upper bound on the value of the
parameter $\varepsilon$, a restriction on the extent of the inertial 
range $[\kappa_1,\kappa_2]$, and in the case of 
\eqref{Eqn:UpperBoundE}\eqref{Eqn:TmeAverageE}, 
an upper bound on the time interval $[0,T_0]$ over which spectral
behavior may occur for a solution. 

There is a well developed literature on the energy transfer rate 
$\varepsilon$, and other aspects of the Kolmogorov -- Obukhov theory,
based on physical assumptions on the character of the fluid motion. 
These assumptions are for flows exhibiting fully developed
turbulence, and are described in Obukhov~\cite{Obukhov1941}, for 
example. They include the hypothesis that the flow is in a 
stochastically steady state, the energy transfer rate is of a 
certain form and exhibits a particular scale invariance, and 
that the energy spectral function is negligible for wave numbers
higher than a cutoff $\kappa_\nu$. Under these assumptions, the 
cutoff $\kappa_\nu$ is determined (it is known as the Kolmogorov 
scale) and the energy transfer rate $\varepsilon$ is identified 
with the energy dissipation rate 
\begin{equation}\label{Eqn:DissipationRatePerVolume}
   \varepsilon_1 := 
   \frac{\nu}{(2\pi)^3}\int_0^{+\infty} \kappa^2 E(\kappa,t) \, d\kappa~.
\end{equation}
Using the latter statement, that $\varepsilon = \varepsilon_1$, it is
possible to find a better estimate of the behavior of $\varepsilon$ 
with respect to Reynolds' number than ours in this paper, as for example in 
C.~Doering \& C.~Foias~\cite{DoeringFoias02}. The difference between this 
body of work and our analysis is that we make no physical assumptions on 
solutions of the Navier -- Stokes equations, deducing our conclusions 
purely from known results about such flows. It is worthwhile to point 
out as well that the upper and lower bounds to the inertial
range in our work are conclusions, as compared to prior work in which
the upper bound $\kappa_\nu$ on the inertial range is an assumption of 
the theory, and no explicit lower bound is given. 

In Section \ref{section2} we give a statement and the proofs of 
our estimates on the Fourier transform of weak solutions of the
Navier -- Stokes equations, posed either on all of $x \in {\mathbb R}^3$
or else for $x \in {\mathbb T}^3$. Since there is no known uniqueness
result, one cannot speak of the solution map for Navier -- Stokes
flow, and we emphasize that this estimate is valid for any weak
solution that satisfies the energy inequality. 
In section \ref{section3} we interpret 
these estimates in the context of the spectral energy function, and we
analyse the constraints on spectral behavior of solutions mentioned
above, giving specific and dimensionally appropriate estimates for the
endpoints of the inertial range $\kappa_1$, $\kappa_2$. In the case
of no inhomogeneous forces, we give an upper bound $T_0$ on the time 
of validity of the spectral regime. The bounds on $\kappa_1$ and $\kappa_2$
are also valid in the probabilistic setting, for statistical ensembles
of solutions. That is, suppose that one is given an ergodic probability 
measure $({\sf P},{\mathcal M})$ on the space of divergence free vector 
fields which is invariant under some choice of definition of 
Navier -- Stokes flow. As long as the inhomogeneous force and 
the support of the invariant probability measure are contained 
in the closure of the set of reasonably smooth divergence free 
vector fields, then our constraints on the spectral behavior of
solutions apply. The final section gives a comparison of our 
constraints on $(\kappa_1,\kappa_2,T_0)$ to the Kolmogorov length 
and time-scales of the classical theory, and a discussion of the 
dimensionless parameter $r_\nu := \kappa_2/\kappa_1$ as an
indicator of spectral behavior of solutions. 

\section{Estimates on the Fourier transform in $L^\infty$}
\label{section2}
The incompressible Navier -- Stokes equations in their usual form are
written for the velocity field $u(x,t)$ of a fluid, its pressure $p(x,t)$,
and a divergence-free force $f(x,t)$,
\begin{eqnarray}\label{Eqn:NavierStokes}
   && \partial_t u + (u \cdot \nabla )u 
   = -\nabla p + \nu \Delta u + f \nonumber \\
   && \nabla \cdot u = 0 ~,             \\
   && u(\cdot,0) = u_0(\cdot) ~,  \nonumber
\end{eqnarray}
where we consider spatial domains either all of Euclidian space 
$x \in {\mathbb R}^3$, or else the compact and boundaryless torus
$x \in {\mathbb T}^3 := {\mathbb R}^3/\Gamma$, where 
$\Gamma \subseteq {\mathbb R}^3$ is a lattice of full rank.
Denote by $D$ either of the above spatial domains. The time domain 
is $0 < t < +\infty$, and the inhomogeneous force function $f$ 
is assumed to be divergence-free and to satisfy 
$f \in L^\infty_{loc}([0,+\infty); H^{-1}(D)\cap L^2(D))$. 
A {\it `Leray' weak solution} to \eqref{Eqn:NavierStokes} 
on $D \times [0,+\infty)$ satisfies the three conditions. 
\begin{enumerate}

  \item {\it Integrability conditions:} For any $T>0$ the 
  vector function $(u,p)$ lies in the following function spaces
  \begin{eqnarray}\label{Eq:Integrebility-u}
     && u \in L^\infty([0,T);L^2(D)) \, \cap \, L^2([0,T); \dot H^1(D)) ~,
     \\ 
     \label{Eq:Integrability-p}
     && p \in L^{5/3}(D \times [0,T))  ~,
  \end{eqnarray}

  \item {\it Weak solution of the equation:} the pair $(u,p)$ is a 
       distributional solution of \eqref{Eqn:NavierStokes}, and 
       furthermore $\lim_{t \to 0^+} u(\cdot,t) =
       u_0(\cdot)$ exists in the strong $L^2$ sense,  

  \item {\it Energy inequality:} the energy inequality is satisfied
  \begin{eqnarray}\label{Eq:EnergyInequality}
    && \Half\int_D |u(x,t)|^2 \, dx 
   + \nu\int_0^t \int_D |\nabla u(x,s)|^2 \, dxds  \\
    && \qquad - \int_0^t \int_D u(x,s)\cdot f(x,s) \, dxds   
     \leq \Half\int_D |u_0(x)|^2 \, dx   \nonumber 
  \end{eqnarray}

  \end{enumerate}
for all $0 < t < +\infty$. The inequality \eqref{Eq:EnergyInequality} is
an identity for solutions which are regular. It is well known that 
weak solutions exist globally in time, either when $f=0$, a result 
due to Leray \cite{Leray34a,Leray34b}, or when $f$ is nonzero.
The question of their uniqueness and regularity remains open. 

Many facts are known about weak solutions, including that for any 
$T>0$ the interpolation inequalities hold; $u \in L^s([0,T);L^p(D))$ 
for all $3/p+2/s=3/2$, for $2\leq p \leq 6$. That the $L^{5/3}$ estimate for
the pressure in \eqref{Eq:Integrability-p} is sufficient 
is due to \cite{SohrVonWahl}. Considering a weak solution as 
a curve in $L^2(D)$ defined over $t \in \mathbb{R}^+$, the 
following proposition holds.
\begin{proposition} 
A weak solution is a mapping $[0,T) \mapsto L^2(D)$ satisfying the 
continuity properties
\begin{equation}\label{Eqn:ContinuityProperties}
  u(\cdot) \in L^\infty([0,T);L^2(D)) \, \cap \, 
    C([0,T);L^2(D)_{\hbox{\small weak topology}}) \,
    \cap \, C([0,T);H^{-\delta}(D))
\end{equation}
for any $\delta >0$. Furthermore, as a curve in Sobolev space,  
\[
   \frac{du}{dt} \in L^2([0,T);H^{-3/2}(D))~. 
\]
\end{proposition}
A clear exposition which includes these basic facts is the lecture
notes of J.-Y. Chemin~\cite{Chemin}. 

Being a curve in $L^2(D)$, the Fourier transform of a weak
solution makes sense, and $\hat{u}(\cdot,t) = {\mathcal F}u(\cdot,t)$ 
is again a curve in $L^\infty([0,T);L^2_k) \cap 
C([0,T);L^2_{k;\hbox{\small weak topology}})$. 
We will make use of a dimensionally adapted Fourier transform
${\mathcal F}$, namely 
\begin{equation}
   \hat{u}(k) = ({\mathcal F}u)(k) := 
     \frac{1}{V^{1/2}} \int_D e^{-ik\cdot x} u(x) \, dx ~,
\end{equation}
where $k \in {\mathbb R}^3$ when the spatial domain is 
$D = {\mathbb R}^3$, and we set $V=(2\pi)^3$ in standard units 
of volume in ${\mathbb R}^3$. In this setting, the norm of the 
Fourier transform is given by 
$\|\hat u\|^2 = \int_{\mathbb{R}^3} |\hat u(\xi)|^2 d\xi$.
When the domain is 
$D = {\mathbb T}^3 = {\mathbb R}^3/\Gamma$, we take 
$k \in \Gamma'$ the lattice dual to $\Gamma$, we set 
$V = |\Gamma| := \hbox{\rm vol}({\mathbb R}^3/\Gamma)$, and we
define $\|\hat{u}\|^2 := \sum_{k \in \Gamma'} |\hat{u}(k)|^2 |\Gamma'|$. 
With this choice, the Plancherel identity reads
\begin{equation}
   \|u\|^2 = \frac{V}{(2\pi)^3} \|\hat{u}\|^2 ~. 
\end{equation} 
With respect to the normalization, the function $u(x,t)$ has 
units of velocity $L/T$, and its Fourier transform is such that
$|\hat{u}(k,t)|^2$ has units of Fourier space energy density $(L/T)^2L^3$.

\subsection{An estimate on ${\mathcal F}\partial_x u(\cdot, t)$ on the torus}
\label{subsection2.1}
Focus the discussion on the case of the spatial domain $D =
\mathbb{T}^3$. Then any choice of initial data 
$u_0(x) \in L^2(\mathbb{T}^3)$ has uniformly bounded Fourier 
coefficients, indeed $|\hat{u}_0(k)| \leq \|u_0\|$. 
Furthermore, since the complex exponential $e^{ik\cdot x}$ is a
perfectly good element of $(H^{-3/2})^*$ which, being tested against 
$u(\cdot,t)$ gives the Fourier coefficients, we also have the
result
\begin{proposition}
For each $k \in \Gamma'$ the Fourier coefficient 
$\hat{u}(k,t) \in {\mathbb C}^3$ 
is a Lipschitz function of $t\in \mathbb{R}^+$. 
\end{proposition} 
It is again made clear by this that the problem of singularity
formation is not that $\hat{u}(k,t)$ becomes unbounded, but rather
that $H^1$ mass, including possibly $L^2$ mass, is propagated to
infinity in $k$-space in finite time. 

A (future) invariant set $A$ is one such that $u_0 \in A$ implies 
for all $t > 0$, $u(t) \in A$ as well. When $f = 0$ the energy 
inequality \eqref{Eq:EnergyInequality} can be viewed as implying 
that the ball $B_R(0) \subseteq L^2(D)$ is an invariant set for 
weak solutions satisfying $\| u_0\|_{L^2(D)} \leq R$. In similar 
terms, we give another invariant set for weak solutions.
Define the set 
\begin{equation}
   A_{R_1} := \{u \, : \, \forall k \in \Gamma', 
     |k||\hat{u}(k)| \leq R_1 \} ~,
\end{equation}
and as above let $B_R(0)$ denote the ball of radius $R$ in 
$L^2(\mathbb{T}^3)$.
\begin{theorem}\label{Thm:one}
In the case that $f=0$, whenever 
\begin{equation}\label{Eq:EstimateOnA}
   \frac{R^2}{\sqrt{V}} \leq \nu R_1
\end{equation}
then the set $A_{R_1}\cap B_R(0)$ is invariant for weak solutions 
of the Navier -- Stokes equations~\eqref{Eqn:NavierStokes}. Thus,
if the initial data satisfies $u_0 \in A_{R_1}\cap B_R(0)$, for
all $0 < t < +\infty$ the Fourier coefficients of any Leray weak
solution emanating from this data satisfy
\begin{equation}
   \sup_{0<t<+\infty} |\hat{u}(k,t)| \leq \frac{R_1}{|k|} ~, \qquad 
   \forall k \in \Gamma' ~.
\end{equation}
\end{theorem}

This result appears in the paper \cite{BiryukCraigIbrahim06} in 
a slightly different form.  
For nonzero $f$ the ball $B_R(0) \subseteq L^2(D)$ is not 
necessarily invariant. However given $u_0 \in B_R(0)$ and our 
hypothesis that $f \in L^\infty_{loc}([0,+\infty); H^{-1}(D)\cap L^2(D))$, 
there is always a nondecreasing function 
$R(T)\geq R$ such that for all $T>0$, $u(\cdot,T) \in B_{R(T)}(0)$. 
Indeed, suppose that a Galilean frame is chosen and the pressure 
$p$ is suitably normalized so that $\int_D u(x,T) \, dx = 0 = 
\int_D f(x,T) \, dx$. Let 
$F^2(T) := \int_0^T \| f \|^2_{\dot{H}^{-1}} \, dt$. 
Then by standard interpolation one has that 
\begin{equation}\label{Eqn:EnergyInequalityUpperBounds}
  \| u(\cdot, T) \|^2_{L^2} + \nu \int_0^T 
     \|\nabla u(\cdot, s)\|^2_{L^2} \, ds 
   \leq R^2(T) ~.     
\end{equation}
The function $R^2(T)$ is an upper bound for the LHS of the 
energy inequality, for which there is an estimate 
$R^2(T) \leq R^2 + \frac{1}{\nu} F^2(T)$.
In case of a bounded inhomogeneous force 
$f(\cdot,t) \in L^\infty({\mathbb R}^+ ; H^{-1}\cap L^2)$, 
there is an upper estimate $F^2(T) \leq C T$, so that $R^2(T)$ 
exhibits (not more than) linear growth in $T$. 

\begin{theorem}\label{Thm:two}
In the case of nonzero $f(x,t)$, let $R(t)$ be an a priori upper 
bound for $\| u(\cdot,t)\|_{L^2}$, for example as in 
\eqref{Eqn:EnergyInequalityUpperBounds}. Suppose that $R_1(t)$ 
is a nondecreasing function such that for all $(k,t)$ we have
\begin{equation}\label{Eq:EstimateOnAandF}
    \frac{R^2(t)}{\sqrt{V}} 
  + \frac{|\hat{f}(k,t)|}{|k|} <  \nu R_1(t)~,
\end{equation}
then the set 
$\{ (u,t) \ : \ 0<t \, , \ u(\cdot,t) \in A_{R_1(t)}\cap B_{R(t)}(0)\}$ 
is invariant for weak solutions of the equations
\eqref{Eqn:NavierStokes}.
That is, if the initial data satisfies $u_0 \in A_{R_1(0)}\cap B_{R(0)}(0)$
then for all $0 < t < +\infty$ the Fourier coefficients of any weak
solution emanating from $u_0$ and being subject to the force $f$ will 
obey the estimate
\begin{equation}
    |\hat{u}(k,t)| \leq \frac{R_1(t)}{|k|} ~.
\end{equation}
\end{theorem}

It is a common situation for the inhomogeneous force to have properties 
of recurrence, such as if it were time-periodic, or if given by a 
statistical process which is stationary with respect to time. For 
bounded $f$ as above, the estimate $F^2(T) \leq C T$ holds. 
Furthermore, one is interested in those solutions which are themselves
statistically stationary. For these solutions it will be the case that
the force adds to the total energy at the same rate as the dissipation 
depletes it. Indeed, one assumes that there is a constant $\overline{R}$ 
such that $\| u(\cdot, T) \|^2_{L^2} \leq \overline{R}^2$ and
$\nu \int_0^T \|\nabla u(\cdot, s)\|^2_{L^2} \, ds \leq \overline{R}^2T$.
Therefore, in the statistically stationary case we expect that the
upper bound $R(t)$ in the hypotheses of Theorem \ref{Thm:two} to be 
given by the constant $\overline{R}$, which in particular gives a 
uniform bound in time. In this situation, the constant $R_1$ is also 
time independent. 

\begin{proof}[Proof of Theorems \ref{Thm:one} and \ref{Thm:two}]
 For each $k \in \Gamma'$ the vector $\hat{u}(k) \in \mathbb{C}^2_k
\subseteq \mathbb{C}^3$, where 
$\mathbb{C}^2_k= \{w \in \mathbb{C}^3 \ : \ w\perp k=0 \}$ 
is specified by the divergence-free condition. Because $(u,p)$ is 
a distributional solution, the Fourier coefficients satisfy
\begin{eqnarray}\label{Eqn:FourierTransformNavierStokes}
   \partial_t\hat{u}(k,t) & = & -\nu |k|^2\hat{u}(k,t) 
     - i \Pi_k k \cdot \frac{1}{\sqrt{V}}
         \sum_{k_1} \hat{u}(k-k_1) \otimes \hat{u}(k_1) 
     + \hat{f}(k,t) ~,   \\
   & := & X(u)_k ~,   \nonumber 
\end{eqnarray}
at least in the weak sense, after testing with a smooth cutoff
function $\varphi(t) \in C^\infty$.
We use the notation that $X(u)_k$ is the $k^{th}$ component of the vector
field represented by the RHS. The convolution has introduced the
factor $1/\sqrt{V}$. The operator 
$\Pi_k :\mathbb{C}^{3}\to \mathbb{C}^2_k$ is given by
\[
    \Pi_k(z)=z-(z \cdot k)\frac{k}{|k|^2} ~.
\] 
The Leray projector onto the divergence-free distributional vector 
fields, considered in Fourier space coordinates, is the direct sum of 
the $\Pi_k$. 

The radial component of the vector field 
$X(u)_k$ in $\mathbb{C}^2_k \subseteq \mathbb{C}^3$ is expressed by 
$\hbox{\rm re}\, (\overline{\hat{u}(k)} \cdot X(u)_k)/|\hat{u}(k)|$.
Consider first Theorem~\ref{Thm:one}, which is the case that $f=0$. 
Suppose that $|k||\hat{u}(k)| = R_1$ for some $k$, that is, the 
solution is on the boundary of the region $A_{R_1}$. Since 
$|\Pi_k k \cdot \sum_{k_1} \hat{u}(k-k_1) \otimes \hat{u}(k_1)| 
\leq |k|\|u\|^2_{L^2}$, an estimate of the radial component of 
$X(u)_k$ is that 
\begin{eqnarray}
  \hbox{\rm re}\, (\overline{\hat{u}(k)} \cdot X(u)_k) 
  & = & -\nu |k|^2 |\hat{u}(k)|^2   
    + \frac{1}{\sqrt{V}} \hbox{\rm im} \, \bigl(\overline{\hat{u}(k)} \
       k \Pi_k\sum_{k_1} \hat{u}(k-k_1) \otimes \hat{u}(k_1)\bigr)
   \nonumber \\
   & \leq & -\nu R_1^2 
     + \frac{1}{\sqrt{V}}\| u(\cdot, t)\|^2_{L^2} R_1
     ~. \nonumber 
\end{eqnarray}
When $u(\cdot,t) \in B_R(0)$ and $R^2 < \nu R_1\sqrt{V}$ 
the RHS is negative, implying that integral curves $\hat{u}(k,t)$ 
cannot exit the region. Thus the ball of radius $R_1/|k|$ in 
$\mathbb{C}^2_k$ is a trapping set, or a future invariant set,  
for the vector field $X(u)_k$.  

In case of the presence of a force $f$, suppose again that
$|k||\hat{u}(k,t)| = R_1(t)$ for some $(k,t)$. The radial 
component of $X(u)_k$ satisfies
\begin{eqnarray}\label{Eqn:RadialComponentF}
  \hbox{\rm re}\, (\overline{\hat{u}(k)} \cdot X(u)_k) 
  & = & -\nu |k|^2 |\hat{u}(k)|^2   
    + \frac{1}{\sqrt{V}} \hbox{\rm im} \, \bigl(\overline{\hat{u}(k)} \
       k \Pi_k\sum_{k_1} \hat{u}(k-k_1) \otimes \hat{u}(k_1)\bigr) 
   \nonumber   \\
    & & + \, \hbox{\rm re} \, (\overline{\hat{u}(k)} \cdot \hat{f}(k,t)) 
               \\
   & \leq & -\nu R_1^2 
     + \frac{1}{\sqrt{V}}\| u(\cdot, t)\|^2_{L^2} R_1
     + |\hat{f}(k,t)|\frac{R_1}{|k|} 
     ~. \nonumber 
\end{eqnarray}
Furthermore, the energy $\| u(\cdot, t)\|^2_{L^2}$ is bounded 
by $R^2(t)$. As long as the radial component of $X(u)_k$ (namely 
the quantity in \eqref{Eqn:RadialComponentF} normalized by the 
length $|\hat{u}(k)|= R_1(t)/|k|$) is bounded above by 
the growth rate of the ball itself, namely by $\dot{R_1}/|k|$,  
then solution curves $(u(\cdot, t), t) $ do not exit the set 
$\{ (u(\cdot),t) \ : \ 0<t \, , \ u(\cdot,t) \in 
A_{R_1(t)}\cap B_{R(t)}(0)\}$. In particular this happens 
for nondecreasing $R_1(t)$ whenever 
$(\sqrt{V})^{-1}R^2(t) + |\hat{f}(k,t)|/|k| < \nu R_1(t)$. 

Therefore when the initial data $u_0 \in A_{R_1(0)}$ and 
the force $f$ satisfies \eqref{Eq:EstimateOnAandF}, then the
solution satisfies $u(x,t) \in A_{R_1(t)}$ for $R_1(t)$ finite,  
for all positive times $t$. 
\end{proof}

It is natural to ask what constraints are imposed on the data 
$u_0$ by the condition \eqref{Eq:EstimateOnA}. 
Given smooth initial data $u_0$  and a force $f$
satisfying $|\hat{f}(k,t)| < F_2|k|$, the constant 
$\nu R_1$ can always be chosen so as to satisfy 
\eqref{Eq:EstimateOnA} (the case $f=0$), or if 
the force $f \not= 0$, the function $R_1(t)$ can
be chosen to be nondecreasing and to satisfy 
\eqref{Eq:EstimateOnAandF}. Thus the hypotheses to this 
theorem encompass any reasonable smooth class of initial 
data and inhomogeneous forcing terms. We note that the 
constant $R_1$ scales dimensionally in terms of the 
units $L^{3/2}/T$. 

Under changes of scale, the quantity 
$\sup_{t}\sup_{k} |k| |\hat{u}(k,t)|$ transforms like the 
\textsl{BV-norm} $\sup_{t} \|\partial_x u(\cdot, t)\|_{L^1}$,
and indeed the latter being finite implies the former. 
However as far as we know there is no known global 
bound on the \textsl{BV-norm} of weak solutions to
\eqref{Eqn:NavierStokes}. A related inequality appears
in \cite{Constantin90}, which is a global upper bound
on the $L^1$\textsl{-norm} of the vorticity 
$\omega := \nabla \times u$, again uniformly in time. 

A corollary to this result gives a stronger estimate for 
time integrals of the Fourier coefficients of weak solutions.

\begin{theorem}\label{Thm:TimeAverageEnergySpectrum}
Suppose that a weak solution $u(x,t)$ satisfies 
\eqref{Eqn:EnergyInequalityUpperBounds}, with its initial 
conditions satisfying $u_0 \in A_{R_1}$. Then time integrals 
of the Fourier coefficients obey the stronger estimate
\begin{equation}\label{Eqn:EstimateTimeIntegrals}
   \int_0^T |\hat{u}(k,t)|^2 \, dt \leq \frac{R_2^2(T)}{\nu|k|^4} ~.
\end{equation}
The constant is given by
\[
   R_2(T) = \Half\left( R_4(T) + \sqrt{2R_1^2(0) + R_4^2(T)} \right) ~, 
\] 
where 
\[
    R_4(T) = \frac{R^2(T)}{\nu\sqrt{V}}
       + \frac{F_1(k,T)}{\sqrt{\nu}} ~, 
    \qquad
    F_1(k,T) = \Bigl(\int_0^T |\hat{f}(k,t)|^2 \, dt \Bigr)^{1/2} ~.  
\]
Assuming that $\sup_{k\in {\mathbb Z}^3} F_1(k,T) := F_\infty(T) <
+\infty$, the constant $R_2(T)$ will be independent of $k$. 
\end{theorem}

When the force $f=0$, the constants $R$ and $R_1$ can be taken 
independent of $T$, implying that $R_2$ is also constant in time,
and the estimate \eqref{Eqn:EstimateTimeIntegrals} holds 
uniformly over $0<T<+\infty$. For nonzero forces which 
are $L^\infty$ with respect to time, there is an upper bound 
$F_1(k,T)\sim \sqrt{T}$, and $R_2(T)$ will grow at most linearly 
in time for large $T$. In the situation in which the forcing 
is given by a stationary process, it is expected (but not 
proven at this point in time) that for typical solutions, the
quantity $R_2^2(T)/T$ will have a limit $\overline{R}_2^2$ for 
large time $T$, representing a  balance between energy input 
and  dissipation. Notice that $R_2$ scales dimensionally in 
terms of $L^{3/2}/T$. 

\begin{proof}
Because the field $u(\cdot,t)$ is divergence-free,
$k\cdot\hat{u}(k,\cdot) = 0$, implying the vector identity
$k\cdot\hat{u}(k-k_1)\otimes\hat{u}(k_1) = \hat{u}(k-k_1)\cdot
k_1\otimes \hat{u}(k_1)$. The absolute value of $\hat{u}$ can 
be estimated from \eqref{Eqn:FourierTransformNavierStokes},
\begin{eqnarray*}
  && \Half\partial_t|\hat{u}(k,t)|^2 + \nu|k|^2|\hat{u}(k,t)|^2
  = \hbox{\rm im}\, \frac{1}{\sqrt{V}} \bigl( \overline{\hat{u}(k,t)}\cdot 
      \Pi_k \sum_{k_1} \hat{u}(k-k_1,t) \cdot k_1\otimes\hat{u}(k_1,t)
      \bigr) \\
    && \qquad + \overline{\hat{u}(k,t)}\cdot \hat{f}(k,t) ~,
\end{eqnarray*}
which is valid  for each $k$ in the sense of weak solutions 
in time. When integrated over the time interval $[0,T]$ it gives
\begin{eqnarray}\label{Eqn:FourierLocalEnergyIdentity}
  &&  \nu |k|^2 \int_0^T |\hat{u}(k,t)|^2 \, dt 
    = \Half|\hat{u}_0(k)|^2 - \Half|\hat{u}(k,T)|^2
   \\
  && \quad + \hbox{\rm im}\, \frac{1}{\sqrt{V}}
      \int_0^T \bigl( \overline{\hat{u}(k,t)}\cdot 
      \Pi_k \sum_{k_1} \hat{u}(k-k_1,t) \cdot k_1\otimes\hat{u}(k_1,t) \bigr)
      \, dt    \nonumber       \\
  &&  \quad + \int_0^T \overline{\hat{u}(k,t)}\cdot \hat{f}(k,t) \, dt ~.
   \nonumber
\end{eqnarray}
Multiplying this identity by $|k|^2$, the terms of the RHS are 
then bounded as follows:
\begin{eqnarray*}
  && \Half|k|^2|\hat{u}_0(k)|^2 \leq \Half R_1^2(0) ~, \\ 
  && |k|^2 \bigl| \int_0^T \overline{\hat{u}(k,t)}\cdot \hat{f}(k,t) 
     \, dt \bigr|
   \leq \frac{1}{\sqrt{\nu}}\bigl(\nu|k|^4
   \int_0^T |\hat{u}(k,t)|^2 \, dt \bigr)^{1/2}
   \bigl( \int_0^T |\hat{f}(k,t)|^2 \, dt  \bigr)^{1/2} ~,
\end{eqnarray*}
and finally 
\begin{eqnarray*}
 && |k|^2 \frac{1}{\sqrt{V}}
      \bigl| \hbox{\rm im}\, \int_0^T \bigl( \overline{\hat{u}(k,t)}\cdot 
      \Pi_k \sum_{k_1} \hat{u}(k-k_1,t) \cdot k_1\otimes\hat{u}(k_1,t) \bigr)
      \, dt  \bigr|   \\ 
 && \leq \frac{1}{\nu \sqrt{V}}\bigl(\nu|k|^4
      \int_0^T |\hat{u}(k,t)|^2 \, dt \bigr)^{1/2} 
      \sup_{0<t<T}\|u(\cdot,t)\|_{L^2} 
      \bigl(\nu\int_0^T |\nabla u(\cdot, s)|^2 \, ds \bigr)^{1/2} ~.
   \nonumber 
\end{eqnarray*}
Define $I^2(k,T) = \nu |k|^4\int_0^T |\hat{u}(k,t)|^2 \,dt$,
then the identity \eqref{Eqn:FourierLocalEnergyIdentity} implies the 
inequality for $I(k,T)$
\begin{equation}
   I^2(k,T) \leq \frac{R_1^2(0)}{2} 
   + \Bigl( \frac{R^2(T)}{\nu\sqrt{V}} 
   + \frac{F_1(k,T)}{\sqrt{\nu}}\Bigr) I(k,T) ~.
\end{equation}
where we have used that $R(t)$ is nondecreasing. The quantity 
$I(k,T)$, being nonnegative, cannot exceed the largest
positive root of the quadratic equation where equality is 
attained, giving the estimate \eqref{Eqn:EstimateTimeIntegrals}.
This estimate is uniform in $k\in{\mathbb Z}^3$ as long as 
$\sup_{k\in{\mathbb Z}^3} F_1(k,T) = F_\infty(T) < +\infty$. 
We note that 
\[
   F^2(T)=\sum_{k\in {\mathbb Z}^3\backslash \{0\}} \frac{F_1^2(k,T)}{|k|^2}~,
\]
which appears in the discussion of energy estimate bounds 
\eqref{Eqn:EnergyInequalityUpperBounds}.
\end{proof}

\subsection{The analogous estimate on $\mathbb{R}^3$}
\label{subsection2.2}
Suppose that $\|u(\cdot,t)\|_{L^2} \leq R(t)$ 
(if there is no force, then $R(t) = R(0)$ suffices). 
The main difference in the case of $D = {\mathbb R}^3$ is that the
functions $\hat{u}(k,t)$ are elements of a Hilbert space, whose 
values at a particular Fourier space-time point $(k,t)$ are not
well defined. We work instead with filtered values of the vector
field $u(x,t)$. Let $0 \not= k \in {\mathbb R}^3$, and for 
$\delta < |k|/(2\sqrt{3})$ define $\hat{\chi}_k(\xi)$ a smooth 
cutoff function of the cube $Q_k$ about $k \in{\mathbb R}^3$ of 
side length $2\delta$ ($\delta \leq 1$ is acceptable for large
$|k|$) which takes value $\hat{\chi}_k = 1$ on a cube of half the 
sidelength. The point is that for 
$\xi \in \hbox{\rm supp}\,(\hat{\chi}_k)$ then 
$|k|/2 \leq |\xi| \leq 3|k|/2$. 
Now define $(\hat{\chi}_k(D)u)(x,t) = 
{\mathcal F}^{-1}\hat{\chi}_k(\xi)\hat{u}(\xi,t) = (\chi_k * u)(x,t)$.
Since $\chi_k \in H^m$ for all $m$, it and its translations are 
admissible test functions, the statement 
\eqref{Eqn:ContinuityProperties} implies that 
$(\chi_k * u)(x,t)$ is a Lipschitz function of $t$ for each $x$. Define
$e_p(k,t) := (\int |\hat{\chi}_k(\xi)\hat{u}(\xi,t)|^p \, d\xi)^{1/p}$ for
$2 \leq p < +\infty$, the conclusion is the following.

\begin{proposition}
The function $e_p^p(k,t)$ is a Lipschitz function of $t \in {\mathbb R}^+$. 
\end{proposition}

We quantify the Fourier behavior of the force $f$ in similar terms. 
Consider the function $\hat{\chi}_k(D)f(x,t) 
   = {\mathcal F}^{-1}\hat{\chi}_k(\xi)\hat{f}(\xi,t)$ 
and let $f_p(k,t) := \sup_{0\leq s\leq t}
   ( \int (|\hat{\chi}_k(\xi)\hat{f}(\xi,t)|^p / |\xi|^p) \,
   d\xi)^{1/p}$.
Recalling that 
$f \in L^\infty_{loc}([0,+\infty);\dot H^{-1}\cap L^2(D))$
we have $f_2(k,t)\leq \sup_{0\leq s\leq t}\|f(k,s)\|_{\dot H^{-1}}$.
However the fact that $f_p$ is finite is in general additional 
information about the regularity of the forcing.

\begin{theorem}\label{Thm:ContinuumL-inftyEstimate}
Let initial conditions $u_0(x)$ give rise to a weak solution
$u(x,t)$ which satisfies $\|u(\cdot,t)\|_{L^2} \leq R(t)$.
Suppose that there exists a nondecreasing function $R_1(t)$ such 
that for all $2 \leq p < +\infty$ and $t\in {\mathbb R}^+$
\begin{equation}\label{Eqn:ContinuumBoundednessRelation}
  (2\delta)^{3/p}\frac{R^2(t)}{\sqrt{V}} + f_p(k,t) < \frac{\nu}{6} R_1(t) ~,
\end{equation}
where $\delta < |k|/2\sqrt{3}$.
Consider a solution to \eqref{Eqn:NavierStokes} that initially satisfies
$\sup_{2 \leq p < +\infty} e_p(k,0) < R_1(0)/|k|$. Then for all 
$t\in {\mathbb R}^+$
\begin{equation}
   \bigl| \hat{\chi}_k(\xi)\hat{u}(\xi,t) \bigr|_{L^\infty} 
  < \frac{R_1(t)}{|k|} ~.
\end{equation}
\end{theorem}

\begin{theorem}\label{Thm:ContinuumL-2TEstimate}
Suppose a weak solution $u(x,t)$ satisfies 
\eqref{Eqn:EnergyInequalityUpperBounds},
and furthermore ask that 
$|\hat{\chi}_k(\xi)\hat{u}_0(\xi)|_{L^\infty} \leq R_1(0)/|k|$. 
Then for all $T\in{\mathbb R}^+$, 
\begin{equation}
   \int_0^T | \hat{\chi}_k\hat{u}(\cdot,t)|_{L^\infty}^2 \, dt 
   \leq \frac{R_2^2(T)}{\nu |k|^4} ~,
\end{equation}
where the constant $R_2(T)$ is given by
\begin{equation}
   R_2(T) = \Half \Bigl(R_5(T) + \sqrt{4R_1^2(0) + R_5^2(T)} \Bigr) ~, 
\end{equation}
where
\begin{equation}\label{Eqn:DefnOfR5}
   R_5(T) = \frac{2R^2(T)}{\nu\sqrt{V}} 
      + \frac{2F_\infty(T)}{\sqrt{\nu}} ~, \qquad
   F_\infty(T) = \sup_{k\in {\mathbb R}^3\backslash \{0\}}\Bigl( \int_0^T 
      |\hat{\chi}_k(\xi)\hat{f}(\xi,t)|^2_{L^\infty}  
        \, dt \Bigr)^{1/2} ~.
\end{equation}
\end{theorem}

The strategy of the proof of these two results is to give an 
analysis similar to that of Section~\ref{subsection2.1} for 
a uniform bound on $e_p(k,t)$ with the correct behavior in 
the parameter $k$. The first lemma controls the behavior 
of $e_2(k,t)$, pointwise in $t$.  

\begin{lemma}\label{Lemma:E2Bounded}
Suppose that $R_1(t)$ is nondecreasing, and is such that 
for all $t \in {\mathbb R}^+$
\begin{equation}
   (2\delta)^{3/2}\frac{R^2(t)}{\sqrt{V}} 
   + f_2(k,t) < \frac{\nu}{6} R_1(t) ~,
\end{equation}
where $\delta < |k|/2\sqrt{3}$. 
If $e_2(k,0) < R_1(0)/|k|$, then for all $0 < t < +\infty$
\begin{equation}\label{Eqn:EstimateForE2}
   e_2(k,t) \leq \frac{R_1(t)}{|k|} ~.
\end{equation}
\end{lemma}

\begin{proof}
The quantity $e_2^2(k,t)$ satisfies the identity
\begin{eqnarray}\label{Eqn:DiffIdentityForE2}
   \Half \frac{d}{dt} e_2^2(k,t) & = & \Half \partial_t 
      \int |\hat{\chi}_k\hat{u}|^2 \, d\xi \\
  & = & \hbox{\rm re}\, \int \overline{(\hat{\chi}_k(\xi)\hat{u}(\xi,t))} 
     \Bigl(-\nu |\xi|^2(\hat{\chi}_k(\xi)\hat{u}(\xi,t))   \nonumber \\
  & & \quad - i \hat{\chi}_k(\xi) \frac{\sqrt{V}}{(2\pi)^3} 
     \Pi_\xi \int \hat{u}(\xi - \xi_1,t) \cdot \xi_1 \otimes
     \hat{u}(\xi_1,t) \, d\xi_1 \Bigr) \, d\xi   \nonumber \\
  & & + \hbox{\rm re}\, \int \overline{(\hat{\chi}_k(\xi)\hat{u}(\xi,t))} 
     \hat{\chi}_k(\xi)\hat{f}(\xi,t) \, d\xi ~.   \nonumber
\end{eqnarray}
The first term of the RHS is negative, bounded above by
\begin{eqnarray*}
  - \hbox{\rm re}\, \nu\int \overline{(\hat{\chi}_k(\xi)\hat{u}(\xi,t))} 
     |\xi|^2(\hat{\chi}_k(\xi)\hat{u}(\xi,t) )\, d\xi 
  & = & -\nu \int |\xi|^2 |\hat{\chi}_k(\xi)\hat{u}(\xi,t)|^2 \, d\xi \\
  & \leq & -\nu \frac{|k|^2}{4} 
     \int |\hat{\chi}_k(\xi)\hat{u}(\xi,t)|^2 \, d\xi ~,
\end{eqnarray*}
where we recall that 
$|\xi| > (|k| - \sqrt{3}\delta) > |k|/2$ holds for 
$\xi \in \hbox{\rm supp}\, (\hat{\chi}_k)$.
The second term of the RHS of \eqref{Eqn:DiffIdentityForE2} is bounded 
with two applications of the Cauchy -- Schwartz inequality;
\begin{eqnarray*}
  && \Bigl| \hbox{\rm im}\, \int \overline{(\hat{\chi}_k(\xi)\hat{u}(\xi,t))} \ 
       \Bigl( \hat{\chi}_k(\xi) \frac{\sqrt{V}}{(2\pi)^3} 
       \Pi_\xi \int \hat{u}(\xi - \xi_1,t) \cdot \xi_1 \otimes
       \hat{u}(\xi_1,t) \, d\xi_1 \Bigr) \, d\xi \Bigr|  \\
  && \quad \leq \frac{\sqrt{V}}{(2\pi)^3} \| \hat{\chi}_k\hat{u} \|_{L^2} \
       \| \hat{\chi}_k \Pi_\xi \xi \cdot \int \hat{u}(\xi - \xi_1,t)  \otimes
       \hat{u}(\xi_1,t) \, d\xi_1 \|_{L^2}  \\
  && \quad \leq \frac{\sqrt{V}}{(2\pi)^3} \| \hat{\chi}_k\hat{u} \|_{L^2} \
     \| \xi \hat{\chi}_k\|_{L^2} \ \bigl|\int \hat{u}(\xi - \xi_1,t)  \otimes
       \hat{u}(\xi_1,t) \, d\xi_1\bigr|_{L^\infty} ~,
\end{eqnarray*}
where we have used the property of incompressibility that  
$\hat{u}(\xi - \xi_1) \cdot \xi_1 = \xi \cdot \hat{u}(\xi - \xi_1)$. 
Furthermore on the support of $\hat{\chi}_k$, 
$|\xi| \leq 3|k|/2$ therefore
\[
   \| \xi \hat{\chi}_k\|_{L^2} \leq \frac{3|k|}{2} (2\delta)^{3/2} ~, \qquad 
   \bigl|\int \hat{u}(\xi - \xi_1,t)  \otimes
       \hat{u}(\xi_1,t) \, d\xi_1\bigr|_{L^\infty} \leq \| \hat{u} \|_{L^2}^2
       \leq \frac{(2\pi)^3 R^2(t)}{V} ~. 
\]
The third term of the RHS of \eqref{Eqn:DiffIdentityForE2} is not present 
without a force. When there is a force, it admits an upper bound
\begin{eqnarray*}
  && \bigl| \hbox{\rm re}\, \int \overline{(\hat{\chi}_k(\xi)\hat{u}(\xi,t))} \ 
     \hat{\chi}_k(\xi)\hat{f}(\xi,t) \, d\xi \bigr| 
  \leq \| |\xi| \hat{\chi}_k\hat{u}\|_{L^2} \ 
       \| |\xi|^{-1} \hat{\chi}_k\hat{f}(\cdot,t)\|_{L^2}   \\
  && \quad \leq \frac{3|k|}{2} \| \hat{\chi}_k\hat{u} \|_{L^2}
  \|f(\cdot,t)\|_{\dot H^{-1}} ~.
\end{eqnarray*}
An estimate of the RHS is therefore
\begin{equation}\label{Eqn:E-VectorField}
   \hbox{\rm RHS} \leq - \frac{\nu}{4} |k|^2 e^2_2(k,t) 
   + \frac{3(2\delta)^{3/2}}{2} \frac{1}{\sqrt{V}} \ R^2(t) \
     |k| e_2(k,t)
   +  \frac{3}{2} f_2(k,t) \ |k| e_2(k,t) ~.  \nonumber 
\end{equation}
This is the situation from which the proof of Theorem~\ref{Thm:two} 
proceeds. Consider the set $B_{R_1} = \{ e \, : \, e \leq (R_1/|k|) \}$,
and suppose that the inequality holds
\begin{equation}\label{Eqn:ConstantsForContinuum}
     \frac{(2\delta)^{3/2}}{\sqrt{V}} \ R^2(t)
     + f_2(k,t) < \frac{\nu}{6} R_1(t) ~. 
\end{equation}
When $e=e_2$ is on the boundary of $B_{R_1}$, that is when $e_2 = R_1(t)/|k|$, 
then 
\begin{eqnarray*}
  && \hbox{\rm RHS} \leq - \frac{\nu}{4} |k|^2 e_2^2(k,t) 
   + \frac{3(2\delta)^{3/2}}{2} \frac{1}{\sqrt{V}} \ R^2(t) \
     |k| e_2(k,t)
   +  \frac{3}{2} f_2(k,t) \ |k| e_2(k,t)  \\
  && \quad\qquad \leq \Bigl( - \frac{\nu}{6} R_1 
   +  (2\delta)^{3/2} \frac{1}{\sqrt{V}} \ R^2(t)
   + f_2(k,t) \Bigr) \frac{3}{2} R_1 < 0 ~. 
\end{eqnarray*}
That is $\dot e_2(k,t) < 0$, and 
thus $B_{R_1}$ is an attracting set for $e_2(k,t)$. If initially
$e_2(k,0) \leq R_1(0)/|k|$, then for all $t \in {\mathbb R}^+$, 
$e_2(k,t) < R_1(t)/|k|$. This proves the lemma.  \qquad
\end{proof}

\begin{lemma}\label{Lemma:two}
Given $k \in {\mathbb R}^3$, suppose that for some $2 \leq p < +\infty$ 
there is a nonincreasing function $R_1(t)$ which satisfies
\begin{equation}\label{Eqn:ContinuumBoundednessRelationP}
  (2\delta)^{3/p}\frac{R^2(t)}{\sqrt{V}} + f_p(k,t) < \frac{\nu}{6} R_1(t) ~,
\end{equation}
for some $\delta < |k|/2\sqrt{3}$.
If a solution to \eqref{Eqn:NavierStokes} initially satisfies
$e_p(k,0) < R_1(0)/|k|$, then for all $t\in {\mathbb R}^+$
\begin{equation}
   e_p(k,t) < \frac{R_1(t)}{|k|} ~.
\end{equation}
\end{lemma}

\begin{proof}
The principle is to show that the local $L^p$ norms of $\hat{u}(\xi,t)$ 
are bounded, using the same strategy as the proof of Lemma~\ref{Lemma:E2Bounded}. 
Since $e_p^p(k,t)$ is Lipschitz continuous for each $k\in {\mathbb R}^3$, 
one calculates
\begin{eqnarray}\label{Eqn:EstimateOfEP}
   \frac{d}{dt} e_p^p(k,t) & = & \partial_t \int |\hat{\chi}_k\hat{u}|^p
   \, d\xi \\
  & = & \hbox{\rm re}\, \int p \, |\hat{\chi}_k\hat{u}|^{p-2} 
        \overline{(\hat{\chi}_k(\xi)\hat{u}(\xi,t))} 
        \partial_t (\hat{\chi}_k(\xi)\hat{u}(\xi,t)) \, d\xi  \nonumber \\
  & = & \hbox{\rm re}\, \int p \, |\hat{\chi}_k\hat{u}|^{p-2} 
        \overline{(\hat{\chi}_k(\xi)\hat{u}(\xi,t))} 
     \Bigl(-\nu |\xi|^2(\hat{\chi}_k(\xi)\hat{u}(\xi,t))   \nonumber \\
  & & \quad - i \hat{\chi}_k(\xi) \frac{\sqrt{V}}{(2\pi)^3} 
     \Pi_\xi \int \hat{u}(\xi - \xi_1,t) \cdot \xi_1 \otimes
     \hat{u}(\xi_1,t) \, d\xi_1 \Bigr) \, d\xi   \nonumber \\
  & & + \hbox{\rm re}\, \int p \, |\hat{\chi}_k\hat{u}|^{p-2}
     \overline{(\hat{\chi}_k(\xi)\hat{u}(\xi,t))} 
     \hat{\chi}_k(\xi)\hat{f}(\xi,t) \, d\xi ~.   \nonumber
\end{eqnarray}
The first term of the RHS of \eqref{Eqn:EstimateOfEP} is negative, 
\[
   -p\nu \int |\xi|^2|\hat{\chi}_k\hat{u}|^p \, d\xi 
   \leq -p\nu\frac{|k|^2}{4} e_p^p(k,t) ~.
\]
Using the assumptions of the lemma and the H\"older inequality, 
the second term has an estimate
\begin{eqnarray*}
  && \Bigl| \hbox{\rm im}\, \int p \, |\hat{\chi}_k\hat{u}|^{p-2} 
       \overline{(\hat{\chi}_k(\xi)\hat{u}(\xi,t))} \ 
       \Bigl( \hat{\chi}_k(\xi) \frac{\sqrt{V}}{(2\pi)^3} 
       \Pi_\xi \xi \cdot \int \hat{u}(\xi - \xi_1,t) \otimes
       \hat{u}(\xi_1,t) \, d\xi_1 \Bigr) \, d\xi \Bigr|  \\
  && \quad \leq p \Bigl(\int |\hat{\chi}_k\hat{u}|^{p} \, d\xi \Bigr)^{(p-1)/p}
       \Bigl( \int |\xi|^p |\chi_k(\xi)|^p \, d\xi \Bigr)^{1/p}
     \frac{\sqrt{V}}{(2\pi)^3}
       \Bigl| \int \hat{u}(\xi - \xi_1,t) \otimes
       \hat{u}(\xi_1,t) \, d\xi_1 \Bigr|_{L^\infty} \\
  && \quad \leq |k| e_p^{p-1} \frac{3p}{2} (2\delta)^{3/p}
     \frac{1}{\sqrt{V}} R^2(t) ~.
\end{eqnarray*}
The third term of the RHS of \eqref{Eqn:EstimateOfEP} is bounded by
\begin{eqnarray*}
   && \Bigl| \hbox{\rm re}\, \int p \, \bigl( |\hat{\chi}_k\hat{u}|^{p-2} 
       \overline{(\hat{\chi}_k(\xi)\hat{u}(\xi,t))} \bigr)
       \hat{\chi_k}\hat{f} \, d\xi \Bigr|   
     \leq p \Bigl(\int |\hat{\chi}_k\hat{u}|^{p} \, d\xi \Bigr)^{(p-1)/p} 
       \frac{3|k|}{2} \Bigl( \int \frac{|\hat{\chi}_k\hat{f}|^p}{|\xi|^p} \, 
         d\xi \Bigr)^{1/p}  \\
  && \qquad \leq \frac{3p}{2} |k| e_p^{p-1} f_p ~. 
\end{eqnarray*}
An estimate of the RHS of \eqref{Eqn:EstimateOfEP} is thus
\begin{eqnarray*}
   RHS \leq p \Bigl( - \frac{\nu}{4} |k|^2 e^p_p(k,t) 
   + \frac{3(2\delta)^{3/p}}{2} \frac{1}{\sqrt{V}} \ R^2(t) \
     |k| e_p^{p-1}(k,t)
   +  \frac{3}{2} f_p(k,t) |k| e_p^{p-1}(k,t) \Bigr) ~.  \nonumber 
\end{eqnarray*}
Consider again the set $B_{R_1} = \{e \, : \, 0 \leq e \leq R_1/|k| \}$.
When $e=e_p$ is on the boundary, that is when $e_p = R_1(t)/|k|$, then 
\begin{eqnarray*}
  && RHS \leq p \Bigl( -\frac{\nu}{4}\frac{R_1^p}{|k|^{p-2}} 
    + \frac{3(2\delta)^{3/p}}{2}\frac{1}{\sqrt{V}} R^2(t) 
      \frac{R_1^{p-1}}{|k|^{p-2}} 
    + \frac{3}{2} f_p \frac{R_1^{p-1}}{|k|^{p-2}} \Bigr) \\
  && \qquad = \Bigl( -\frac{\nu}{6} R_1
    + (2\delta)^{3/p}\frac{1}{\sqrt{V}} R^2(t) 
    + f_p \Bigr) \frac{3}{2}\frac{p R_1^{p-1}}{|k|^{p-2}} ~.
\end{eqnarray*}
Supposing that \eqref{Eqn:ConstantsForContinuum} holds, the
RHS is negative for $e=e_p$ on the boundary, and the set $B_{R_1}$
is attracting for the quantity $e_p(k,t)$ for $t \in {\mathbb R}^+$.
\qquad 
\end{proof}

These are estimates which are uniform in the parameter $p$. 
We are now prepared to complete the proof of 
Theorem~\ref{Thm:ContinuumL-inftyEstimate}, indeed we note that 
$\lim_{p\to +\infty} e_p(k,t) = |\hat{\chi}_k\hat{u}|_{L^\infty}$.
The quantity $e_p(k,t)$ is given a uniform upper bound in 
Lemma~\ref{Lemma:two} under the stated hypotheses, and hence 
the theorem follows. 

\begin{proof}[Proof of Theorem~\ref{Thm:ContinuumL-2TEstimate}]
Start with the identity in \eqref{Eqn:EstimateOfEP} for $e_p(k,t)$,
which we read as
\begin{equation}\label{Eqn:EstOnEP2}
   \partial_t e_p^2 = \frac{2}{p} e_p^{p(2/p-1)} \partial_t e_p^p ~. 
\end{equation}
Because of the support properties of the cutoff functions
$\hat{\chi}_k(\xi)$, there is the comparison 
$|k|/2 \leq |\xi| \leq 3|k|/2$ on the support of $\hat{\chi}_k$,
thus one has upper and lower bounds
\begin{equation}\label{Eqn:K-LocalEquivalence}
     \frac{|k|^2}{4} \int |\hat{\chi}_k(\xi)\hat{u}(\xi,t)|^p \, d\xi 
   \leq  \int |\xi|^2 |\hat{\chi}_k(\xi)\hat{u}(\xi,t)|^p \, d\xi 
   \leq  \frac{9|k|^2}{4} \int |\hat{\chi}_k(\xi)\hat{u}(\xi,t)|^p \, d\xi~.
\end{equation}
We therefore can rewrite the RHS of \eqref{Eqn:EstOnEP2} as
\begin{eqnarray}\label{Eqn:IntegrableIdentityContinuum}
   && 2 e_p^{p(2/p-1)} 
     \Bigl( -\nu \int |\xi|^2 |\hat{\chi}_k(\xi)\hat{u}(\xi,t)|^p \, d\xi \\
   && \quad + \frac{\sqrt{V}}{(2\pi)^3} \hbox{\rm im}\, 
     \int |\hat{\chi}_k(\xi)\hat{u}(\xi,t)|^{p-2}
     \overline{\hat{\chi}_k(\xi)\hat{u}(\xi,t)}  \nonumber \\
   && \quad\qquad \times \hat{\chi}_k(\xi) 
     \Bigl(\Pi_\xi \int \hat{u}(\xi - \xi_1,t)\cdot 
     \xi_1\otimes \hat{u}(\xi_1,t) \, d\xi_1 \Bigr) \, d\xi  \nonumber
     \\
  && \quad + \hbox{\rm re}\, \int |\hat{\chi}_k(\xi)\hat{u}(\xi,t)|^{p-2}
     \overline{\hat{\chi}_k(\xi)\hat{u}(\xi,t)} 
     \hat{\chi}_k(\xi) \hat{f}(\xi,t) \, d\xi \Bigr) ~.   \nonumber \\
  && := - I_1 + I_2 + I_3 \nonumber
\end{eqnarray}
The first of the three terms of the RHS is
\[
   I_1 = 2\nu \left(\frac{\int |\xi|^2 |\hat{\chi}_k(\xi)\hat{u}(\xi,t)|^p \, d\xi}
    {\int |\hat{\chi}_k(\xi)\hat{u}(\xi,t)|^p \, d\xi} \right)^{1-2/p} 
    \Bigl( \int |\xi|^2 |\hat{\chi}_k(\xi)\hat{u}(\xi,t)|^p \, d\xi
    \Bigr)^{2/p} ~,
\]
which is well-defined because of \eqref{Eqn:K-LocalEquivalence}, is
positive, and through a lower bound will give us the result of the 
theorem. The second term of the RHS is 
\begin{eqnarray*}
   && I_2 = \frac{2\sqrt{V}}{(2\pi)^3}
     \hbox{\rm im}\, \int |\hat{\chi}_k\hat{u}|^{p-2}
     (\overline{\hat{\chi}_k\hat{u}}) \, \hat{\chi}_k 
     \Bigl(\Pi_\xi \int \hat{u}(\xi - \xi_1,t)\cdot 
     \xi_1\otimes \hat{u}(\xi_1,t) \, d\xi_1 \Bigr) \, d\xi \\
   && \qquad \times \Bigl( \int |\hat{\chi}_k(\xi)\hat{u}(\xi,t)|^p 
      \, d\xi \Bigr)^{-1+2/p} ~, 
\end{eqnarray*}
for which one uses the H\"older inequality (with $(p-2)/p + 1/p + 1/p
= 1$) to obtain an upper bound;
\begin{eqnarray*}
  && |I_2| \leq \frac{2\sqrt{V}}{(2\pi)^3}
     \bigl( \int |\hat{\chi}_k\hat{u}|^p \, d\xi \bigr)^{(p-2)/p}
     \bigl( \int |\hat{\chi}_k\hat{u}|^p \, d\xi \bigr)^{1/p}  
     \bigl(\int |\hat{\chi}_k|^p \, d\xi \bigr)^{1/p}  \\ 
  && \qquad \times \Bigl| \Pi_\xi \int \hat{u}(\xi - \xi_1,t)\cdot 
     \xi_1\otimes \hat{u}(\xi_1,t) \, d\xi_1 \Bigr|_{L^\infty} 
       \Bigl( {\int |\hat{\chi}_k\hat{u}|^p \, d\xi}\Bigr)^{-1+2/p} ~.
\end{eqnarray*}
The third term of the RHS is
\begin{eqnarray*}
  && I_3 = 2\hbox{\rm re}\, \int |\hat{\chi}_k(\xi)\hat{u}(\xi,t)|^{p-2}
     \overline{\hat{\chi}_k(\xi)\hat{u}(\xi,t)} 
     \hat{\chi}_k(\xi) \hat{f}(\xi,t) \, d\xi    \\
  && \qquad\quad \times 
     \Bigl( {\int |\hat{\chi}_k\hat{u}|^p \, d\xi}\Bigr)^{-1+2/p} ~,
\end{eqnarray*}
which satisfies an estimate of similar form, namely
\begin{eqnarray*}
  && |I_3| \leq 2\bigl( \int |\hat{\chi}_k\hat{u}|^p \, d\xi \bigr)^{(p-2)/p}
     \bigl( \int |\hat{\chi}_k\hat{u}|^p \, d\xi \bigr)^{1/p} 
     \bigl( \int |\hat{\chi}_k\hat{f}|^p \, d\xi \bigr)^{1/p} 
     \Bigl( {\int |\hat{\chi}_k\hat{u}|^p \, d\xi}\Bigr)^{-1+2/p}  \\
  && \qquad = 2\bigl( \int |\hat{\chi}_k\hat{u}|^p \, d\xi \bigr)^{1/p}
      \bigl( \int |\hat{\chi}_k\hat{f}|^p \, d\xi \bigr)^{1/p} 
      ~.
\end{eqnarray*}
Integrating \eqref{Eqn:EstOnEP2} over the interval $[0,T]$,
\begin{eqnarray}\label{Eqn:IntegratedEstimateOnEP2}
  && 2 \int_0^T \nu \Bigl( \int |\xi|^2 |\hat{\chi}_k(\xi)\hat{u}(\xi,t)|^p 
    \, d\xi \Bigr)^{2/p} \, 
      \left(\frac{\int |\xi|^2 |\hat{\chi}_k(\xi)\hat{u}(\xi,t)|^p \, d\xi}
      {\int |\hat{\chi}_k(\xi)\hat{u}(\xi,t)|^p \, d\xi} \right)^{1-2/p}
    \, dt   \\
  && \quad \leq e_p^2(k,0) - e_p^2(k,T) + \int_0^T |I_2(t)| +
  |I_3(t)|\, dt ~.   \nonumber  
\end{eqnarray}
Because of the properties of $\hat{\chi}_k$, we have
$|k|/2 \leq |\xi| \leq 3|k|/2$ in the support of the integrand, 
and therefore
\[
     |k|^{2-4/p} \leq 
     \left(2 \frac{\int |\xi|^2 |\hat{\chi}_k(\xi)\hat{u}(\xi,t)|^p \, d\xi}
    {\int |\hat{\chi}_k(\xi)\hat{u}(\xi,t)|^p \, d\xi} \right)^{1-2/p}  ~,
\]
which means that the LHS of \eqref{Eqn:IntegratedEstimateOnEP2} gives
an upper bound for the quantity
\[
   \int_0^T \nu |k|^2 \Bigl( 
       \int |\hat{\chi}_k(\xi)\hat{u}(\xi,t)|^p 
       \, d\xi \Bigr)^{2/p} \, dt ~.
\]
Cancelling terms, one uses Cauchy -- Schwartz to estimate the 
two time integrals on the RHS of \eqref{Eqn:IntegratedEstimateOnEP2}. 
Using that 
\[
   \Bigl| \Pi_\xi \int \hat{u}(\xi - \xi_1,t)\cdot 
     \xi_1\otimes \hat{u}(\xi_1,t) \, d\xi_1 \Bigr|_{L^\infty} 
   \leq \| \hat{u}(\cdot,t)\|_{L^2} \| \xi \hat{u}(\cdot,t)\|_{L^2} ~,
\]
we estimate the first time integral as follows:
\begin{eqnarray*}
  && \int_0^T |I_2(t)| \, dt \leq \frac{2\sqrt{V}}{(2\pi)^3}
      \bigl(\int |\hat{\chi}_k|^p \, d\xi \bigr)^{1/p} 
      \Bigl( \int_0^T \bigl( \int |\hat{\chi}_k\hat{u}|^p 
         \, d\xi \bigr)^{2/p} \, dt \Bigr)^{1/2}   \\
  &&\qquad\quad \times \Bigl( \int_0^T \| \hat{u}(\cdot,t)\|_{L^2}^2 
         \| \xi \hat{u}(\cdot,t)\|_{L^2}^2 \, dt \Bigr)^{1/2} \\
  &&\quad \leq \frac{2\sqrt{V}}{(2\pi)^3} (2^3\delta^3)^{1/p}
       \Bigl( \int_0^T \bigl( \int |\hat{\chi}_k\hat{u}|^p 
         \, d\xi \bigr)^{2/p} \, dt \Bigr)^{1/2}  \\
  &&\qquad\quad \times \left( \frac{(2\pi)^3}{\sqrt{\nu}V} (\sup_{0\leq t \leq T}
       \|u(\cdot,t)\|_{L^2}) \bigl( \int_0^T \nu \|\nabla u(\cdot,t)\|_{L^2}^2 
      \, dt \bigr)^{1/2} \right) ~.
\end{eqnarray*}
We have used the Plancherel identity and its constant, as well as the 
fact that $\int|\hat{\chi}_k|^p \, d\xi \leq (2\delta)^3$. Thus
\[
   \int_0^T |I_2| \, dt \leq \frac{2 R^2(T)}{\nu\sqrt{V}} (2^3\delta^3)^{1/p}
     \Bigl( \int_0^T \nu \bigl( \int |\hat{\chi}_k\hat{u}|^p 
         \, d\xi \bigr)^{2/p} \, dt \Bigr)^{1/2} ~.
\] 
Under similar considerations, 
\begin{eqnarray*}
  && \int_0^T |I_3| \, dt \leq 2 \Bigl( \int_0^T
     \bigl( \int |\hat{\chi}_k\hat{u}|^p \, d\xi \bigr)^{2/p} 
     \, dt \Bigr)^{1/2}   
     \Bigl( \int_0^T
     \bigl( \int |\hat{\chi}_k\hat{f}|^p \, d\xi \bigr)^{2/p} 
     \, dt \Bigr)^{1/2}  ~.
\end{eqnarray*}
Now multiply the inequality \eqref{Eqn:IntegratedEstimateOnEP2} by
$|k|^2$ and use the above estimates with the fact that 
$|k|/2 \leq |\xi| \leq 3|k|/2$ on the support of $\hat{\chi}_k$. 
\begin{eqnarray}\label{Eqn:IntegratedInequalityOnEP2}
 &&  \int_0^T \nu |k|^4 \bigl( \int |\hat{\chi}_k\hat{u}|^p
      \, d\xi \bigr)^{2/p} \, dt 
      \leq |k|^2 e_p^2(0) \nonumber \\ 
  &&\qquad\quad  + 2 ((2\delta)^3)^{1/p} \frac{R^2(T)}{\nu\sqrt{V}} 
     \Bigl( \int_0^T \nu |k|^4 \bigl( \int |\hat{\chi}_k\hat{u}|^p 
         \, d\xi \bigr)^{2/p} \, dt \Bigr)^{1/2}   \\
  && \qquad\quad + \frac{2}{\sqrt{\nu}} \Bigl( \int_0^T
     \bigl( \int |\hat{\chi}_k\hat{f}|^p \, d\xi \bigr)^{2/p} 
     \, dt \Bigr)^{1/2} 
     \Bigl( \int_0^T \nu |k|^4
     \bigl( \int |\hat{\chi}_k\hat{u}|^p \, d\xi \bigr)^{2/p} 
     \, dt \Bigr)^{1/2}  ~.    \nonumber
\end{eqnarray}
From our hypotheses on the initial data we know that $|k|^2
e_p^2(0)\leq R_1^2(0)$. Defining 
$I_p^2(k,T) := \int_0^T \nu |k|^4 \bigl( \int |\hat{\chi}_k\hat{u}|^p
\, d\xi \bigr)^{2/p} \, dt$, the inequality 
\eqref{Eqn:IntegratedInequalityOnEP2}  states that
\begin{equation}
   I_p^2(k,T) \leq R_1^2(0) 
   + \Bigl( 2(2\delta)^{3/p} \frac{R^2(T)}{\nu\sqrt{V}}
     + \frac{2F_p(T)}{\sqrt{\nu}} \Bigr) I_p(k,T) ~,
\end{equation}
where we define $F_p(T) := ( \int_0^T
( \int |\hat{\chi}_k\hat{f}|^p \, d\xi )^{2/p} \, dt )^{1/2}$.  
As we have argued before, this implies that $I_p(k,t)$ cannot
exceed the largest positive root $R_{2,p}$ of the associate 
quadratic equation, resulting in the statement that 
\[
   I_p(k,t) \leq R_{2,p}(T)
\]
where the constant $R_{2,p}(T)$ is given by
\begin{equation}
   R_{2,p}(T) = \Half \Bigl(R_{5,p}(T) 
   + \sqrt{4R_1^2(0) + R_{5,p}^2(T)} \Bigr) ~, 
\end{equation}
where in turn
\[
   R_{5,p}(T) = 2(2\delta)^{3/p} \frac{R^2(T)}{\nu\sqrt{V}} 
      + \frac{2F_p(T)}{\sqrt{\nu}} ~, \qquad
   F_p^2(T) = \Bigl( \int_0^T (\int 
      |\hat{\chi}_k(\xi)\hat{f}(\xi,t)|^p \, d\xi )^{2/p}
        \, dt \Bigr) ~.
\]
The result of the theorem will follow by taking the limit of large 
$p \to +\infty$, recovering the estimate on
$|\hat{\chi}_k\hat{u}|_{L^\infty}$. \qquad 
\end{proof}

\section{Estimates of energy spectra}
\label{section3}

The {\it energy  spectral function} is the main concern of the present
paper. For the problem \eqref{Eqn:NavierStokes} posed on 
$D = {\mathbb R}^3$ this is defined by the spherical integrals 
\begin{equation}
   E(\kappa,t) = \int_{|k|=\kappa} |\hat{u}(k,t)|^2 \, dS(k)  ~,
\end{equation}
where $0 \leq \kappa < +\infty$ is the radial coordinate in Fourier
transform variables. When considering the 
case of a periodic domain $D = {\mathbb T}^3$ the Fourier transform 
is defined over the dual lattice, and therefore to avoid questions of
analytic number theory one defines the energy spectral function to be
a sum over Fourier space annuli of given thickness $a$;
\begin{equation}\label{Eqn:EnergySpectralFunction}
   E(\kappa,t) = \frac{1}{a} 
      \sum_{\kappa \leq |k| < \kappa + a} |\hat{u}(k,t)|^2 |\Gamma'| ~.
\end{equation}
The classical Sobolev space norms of the function $u$ can be defined
in terms of the energy spectral function, via the Plancherel identity. 
Indeed in the case $D={\mathbb R}^3$ the $L^2$ norm is given as
\begin{equation}
   \| u\|^2_{L^2} = \frac{V}{(2\pi)^3} 
       \int_0^{+\infty} E(\kappa) \, d\kappa ~,
\end{equation}
and the $H^r$ Sobolev norms are
\begin{equation}
   \| u\|^2_{H^r} = \frac{V}{(2\pi)^3} 
       \int_0^{+\infty} (\kappa^2 + 1)^r E(\kappa) \, d\kappa ~.
\end{equation}
Analogous definitions hold for the case $x \in {\mathbb T}^3$.

\subsection{Kolmogorov spectrum}\label{Section:KSpectrum}
There is considerable lore and a large literature on the behavior of
the spectral function, particularly for large Reynolds number flows,
the most well known statement being due to Kolmogorov. The prediction 
depends upon a parameter $\varepsilon$, which is interpreted 
as the {\it average rate of energy transfer per unit volume}. 
Assuming that a flow exhibiting fully developed
and isotropic turbulence has a regime of  wavenumbers over which 
$E(\kappa, \cdot)$ depends only upon $\varepsilon$ and $\kappa$,  
Kolmogorov's famous argument states that over an {\it inertial range}
$\kappa \in [\kappa_1, \kappa_2]$, 
\begin{equation}\label{Eqn:KolmogorovLaw}
   E(\kappa,\cdot) \sim C_0\varepsilon^{2/3}\kappa^{-5/3} ~,  
\end{equation}
for a universal constant $C_0$. His reasoning is through a 
dimensional analysis. The actual history of this 
prediction, which is well documented in \cite{Frisch1995} among
other references, includes a number of statements of Kolmogorov
as to the small scale structure of the fluctuations in a turbulent
flow \cite{K41a,K41c,K41b}, and an interpretation of his results by Obukhov 
\cite{Obukhov1941} in terms of the Fourier transform, as is stated 
in \eqref{Eqn:KolmogorovLaw}\footnote{In fact Obukhov formulated an
  integral version of this result, which he called the `two-thirds law
  of energy distribution'.}. 
Some of the issues surrounding this
statement are whether the Kolmogorov scaling law
\eqref{Eqn:KolmogorovLaw} should hold for an individual flow  
at every instant in time, whether it should hold on time average,
or whether it is a statement for the average behavior for a 
statistical ensemble of flows with the probability measure for 
this ensemble being given by some natural invariant measure for 
solutions of the Navier -- Stokes  equations. The bounds given
below have implications on the energy spectral function in all of
these cases. 
 
\subsection{Bounds on energy spectra}
The estimates given in section \ref{section2} on the Fourier transform 
of solutions translate into estimates on the energy spectral function
for such solutions. Bounds which are pointwise in time are given in
the following theorem. 

\begin{theorem}\label{Thm:four}
Suppose that $f=0$ and that the initial data satisfies 
$u_0 \in A_{R_1}\cap B_R(0)$, where $R$ and $R_1$ satisfy 
\eqref{Eq:EstimateOnA}. Then for all $\kappa$ and all times 
$t$, 
\begin{equation}\label{Eqn:UpperBound-NoForce}
   E(\kappa,t) \leq 4\pi R_1^2 ~.
\end{equation}
In the case of non-zero forcing $f$, then there is a finite but
possibly growing upper bound given by
\begin{equation}\label{Eqn:UpperBound-Force}
  E(\kappa,t) \leq 4\pi R_1^2(t) 
\end{equation} 
\end{theorem}

Bounds which concern the time average of the energy spectral function
are derived from Theorem \ref{Thm:TimeAverageEnergySpectrum}. 

\begin{theorem}\label{Thm:five}
Again suppose that the initial data satisfies $u_0 \in A_{R_1}\cap B_R(0)$, 
where $R$ and $R_1$ satisfy \eqref{Eq:EstimateOnA}, and the force
$f \in L^\infty_{loc}([0,+\infty); H^{-1}(D)\cap L^2(D))$ is bounded
as in \eqref{Eqn:ContinuumBoundednessRelation}\eqref{Eqn:DefnOfR5}. 
Then for every $T$ the energy spectral function satisfies
\begin{equation}\label{Eqn:EnergySpectrumDecay}
   \frac{1}{T} \int_0^T E(\kappa,t) \, dt 
   \leq \frac{4\pi R_2^2(T)}{\nu T}\frac{1}{\kappa^2} ~.
\end{equation}
\end{theorem}

In particular, under the hypotheses of Theorems~\ref{Thm:one}, 
\ref{Thm:two}, and \ref{Thm:TimeAverageEnergySpectrum}, 
the energy spectrum must decay with an upper bound of order 
${\mathcal O}(\kappa^{-2})$ for every $T$. In case that the solution 
is such that $\limsup_{T\to +\infty}R_2^2(T)/T$ is finite, then
the time average behavior in \eqref{Eqn:EnergySpectrumDecay}
has an upper bound which is uniform in $T$. In any case,
this is evidently faster than the 
Kolmogorov power law \eqref{Eqn:KolmogorovLaw} and thus merits a
further discussion. 

\begin{proof}[Proof of Theorems \ref{Thm:four} and \ref{Thm:five}]
In the case of spatially periodic solutions, the definition of the
energy spectral function gives that
\begin{eqnarray*}
   E(\kappa,t) & = & \frac{1}{a} 
     \sum_{\kappa \leq |k| < \kappa + a} |\hat{u}(k,t)|^2 |\Gamma'|  \\
   & \leq & \frac{1}{a} \sum_{\kappa \leq |k| < \kappa + a} 
          \frac{R_1^2}{|k|^2} |\Gamma'|     \\
   & \leq & \frac{1}{a} \frac{4\pi\kappa^2 a}{|\Gamma'|} 
            \frac{R_1^2}{\kappa^2} |\Gamma'| 
      \leq 4\pi R_1^2(t) ~. 
\end{eqnarray*}
We have used that the lattice point density of $\Gamma'$ 
is $|\Gamma'|^{-1}$. The inequalities of Theorem~\ref{Thm:four}
follow. In the case in which the spatial domain $D = {\mathbb R}^3$, 
the proof is similar. 

To prove Theorem~\ref{Thm:five}, consider first the case of
$D = {\mathbb R}^3$, where
\begin{eqnarray*}
   \frac{1}{T} \int_0^T E(\kappa,t) \, dt & = & 
     \int_{|k|=\kappa} \Bigl( \frac{1}{T} \int_0^T 
     |\hat{u}(k,t)|^2 \, dt \Bigr) dS(k)   \\
   & \leq & 4\pi \kappa^2 \Bigl( \frac{R_2^2(T)}{\nu T\kappa^4} \Bigr) ~.
\end{eqnarray*}
This is the stated estimate. The periodic case is similar. \quad
\end{proof}

\subsection{Estimates on the inertial range}\label{SubSection3.3}
The two theorems \ref{Thm:four} and \ref{Thm:five} have implications
on the inertial range of a solution of \eqref{Eqn:NavierStokes}. 
In particular the inequalities 
\eqref{Eqn:UpperBound-NoForce}\eqref{Eqn:UpperBound-Force}  give 
uniform upper bounds for $E(\kappa,t)$, while 
\eqref{Eqn:EnergySpectrumDecay} estimates its time averages 
from above with a decay rate $C\kappa^{-2}$. For direct comparison 
we define the idealized Kolmogorov energy spectral function with
parameter $\varepsilon$ to be 
\begin{equation}\label{Eqn:IdealizedKolmogorovEnergySpectrum}
   E_K(\kappa) = C_0\varepsilon^{2/3}\kappa^{-5/3} ~; 
\end{equation}
this is to be considered to be stationary in time so that it also
represents the idealized time average. These bounds and the idealized 
energy spectral function are illustrated in Figure~\ref{Fig:1}. 
The first constraint implied by 
\eqref{Eqn:UpperBound-NoForce}\eqref{Eqn:UpperBound-Force} and 
\eqref{Eqn:EnergySpectrumDecay} is that a spectral regime with 
parameter $\varepsilon$ is incompatible with the situation in 
which $E_K(\kappa)$ lies entirely above the permitted set 
$S := \{ E \leq 4\pi R_1^2\} \cap 
\{ E \leq 4\pi R_2^2(T)/\nu\kappa^2 T \}$. 

\begin{figure}[htb]
  \begin{center}
  \epsfig{file=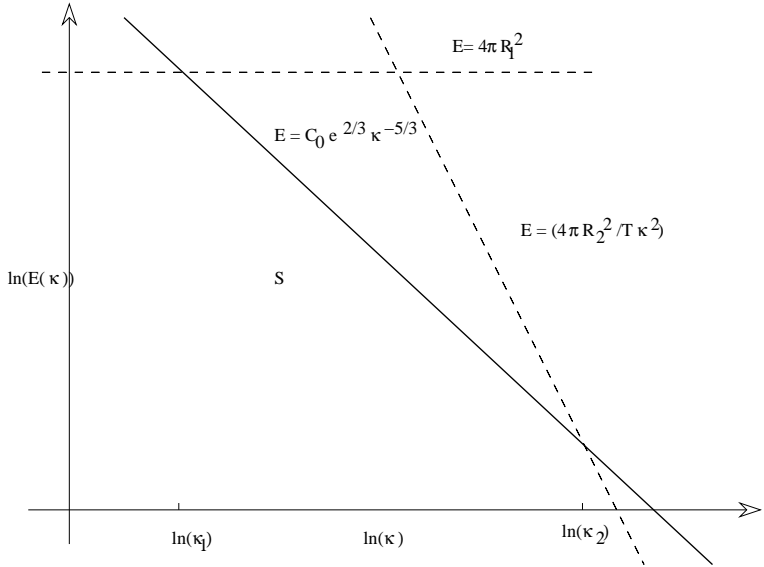}    
  \end{center}
  \caption{The accessible set $S$ and a spectral function $E_K(\kappa)$.}
  \label{Fig:1}
\end{figure}

\begin{proposition}
In order that the graph of $E_K(\kappa)$ intersect the set $S$, the
parameters must satisfy the relation
\begin{equation}\label{Eqn:GraphIntersection}
   \nu^{5/6} C_0\varepsilon^{2/3} \leq 
   4\pi \left(\frac{R_2(T)}{\sqrt{T}}\right)^{5/3} R_1^{1/3}(T) ~.
\end{equation}
\end{proposition}

The proof is elementary. This gives an upper bound on the parameter 
$\varepsilon$, in fact on the quantity
$C_0\nu^{5/6}\varepsilon^{2/3}$, in terms of quantities that are
determined by the initial data and the inhomogeneous forces. In the
setting of statistically stationary solutions, 
$R_2^2(T)/T \leq \overline{R}_2^2$, where $\overline{R}_2$ and 
$R_1$ are constant. In order that a spectral regime is achieved, 
the relation 
\[
    C_0 \nu^{5/6} \varepsilon^{2/3} \leq 4\pi
     \overline{R}_2^{5/3}R_1^{1/3} ~
\]
must hold. This constrains the values of the parameter $\varepsilon$
for any solution regime that exhibits spectral behavior.

We now take up the question of the endpoints of the inertial range
$[\kappa_1,\kappa_2]$, assuming a given value of $\varepsilon$. We 
will produce an interval $[\overline{\kappa}_1,\overline{\kappa}_2]$
such that upper and lower limits of the inertial range, respectively 
$\kappa_1$ and $\kappa_2$ must necessarily satisfy 
$\overline{\kappa}_1 \leq \kappa_1 \leq \kappa_2 \leq \overline{\kappa}_2$.
First of all, the function $E_K(\kappa)$ will violate the estimate 
\eqref{Eqn:UpperBound-NoForce} (if the force is not present) 
or \eqref{Eqn:UpperBound-Force} (when there is a force) unless 
$\kappa \geq \overline{\kappa}_1$, where 
\begin{equation}\label{Eqn:DefinitionKappa1}
   C_0\varepsilon^{2/3}\overline{\kappa}_1^{-5/3} = 4\pi R_1^2 ~, 
\end{equation}
which gives a bound from below for the lower endpoint of the inertial
range. 

\begin{proposition}
An absolute lower bound for the inertial range is given by
\begin{equation}\label{Eqn:BoundsOnKappa1}
   \overline{\kappa}_1 = \frac{C_0^{3/5}\varepsilon^{2/5}}{(4\pi R_1^2)^{3/5}} ~.
\end{equation}
\end{proposition}

It is an amusing exercise to check that the RHS has the appropriate
units of $L^{-1}$, for which we note that the units of $\varepsilon$ 
are $L^2/T^3$. In the case of a nonzero forcing, $R_1(t)$ may be
increasing, in which case $\overline{\kappa}_1(t)$ would decrease. In the case
of bounded forces, $R_1(t)$ may increase linearly in $t$, implying 
that $\overline{\kappa}_1(t) \sim t^{-6/5}$. For a statistically stationary
solution as described above, $R$ and $R_1$, and therefore $\overline{\kappa}_1$
are constant. 

The upper bound for the inertial range comes from comparing time
averages of $E_K$ with the upper bound \eqref{Eqn:EnergySpectrumDecay}. 
Indeed,
\begin{equation}\label{Eqn:UpperBoundForKappa2}
   \frac{1}{T} \int_0^T E_K(\kappa) \, dt 
  = C_0 \varepsilon^{2/3}\kappa^{-5/3}
  \leq \frac{4\pi R_2^2(T)}{\nu T \kappa^2} ~.
\end{equation}

\begin{proposition}
The inequality \eqref{Eqn:UpperBoundForKappa2} holds only over an 
interval of $\kappa$ bounded above by 
\begin{equation}\label{Eqn:BoundsOnKappa2}
   \overline{\kappa}_2 =
   \frac{(4\pi)^3}{(C_0\nu)^3}\frac{1}{\varepsilon^2}\frac{R_2^6(T)}{T^3}
   ~. 
\end{equation}
\end{proposition}

It is again amusing to check that the RHS has units of $L^{-1}$,
noting that $\nu$ has units of $L^2/T$. When there is a force 
present, the constant $R_2^2(T)$ may grow in $T$. When  
$f \in L^\infty({\mathbb R}^+; H^{-1}\cap L^2)$ the constant
$R_2(T)$ grows at most linearly in $T$. When considering the
case of a bounded and statistically stationary forcing term, for
example, the ratio $R_2^2(T)/T$ is expected to have a limit as 
$T$ grows large, $\lim_{T\to +\infty}R_2^2(T)/T = \overline{R}_2^2$,
which gives rise to a fixed upper bound for $\overline{\kappa}_2$. 
Indeed in any case in which the constant 
$\overline{R}_2^2 := \limsup_{T\to +\infty}(R_2^2(T)/T)$ is finite,
this argument gives an upper bound for $\overline{\kappa}_2$.  

However with no force present, or with a force which decays 
in time, then $R_2^2(T)$ will be bounded, or may grow sublinearly, which
results in the bound for $\overline{\kappa}_2 = \overline{\kappa}_2(T)$ 
which is decreasing in time. Supposing that at some time $T_0$ 
we have that for $T > T_0$ then 
$\overline{\kappa}_2(T) \leq \overline{\kappa}_1$, implying that the interval
consisting of the inertial range is necessarily empty. The explicit 
bound for $T_0$ in the case of no force present is as follows.

\begin{proposition}\label{Thm:Prop6}
Suppose that the force $f=0$, so that $R_1$ and $R_2$ are constant 
in time. Then $\overline{\kappa}_2(T) \leq \overline{\kappa}_1$ 
for all $T \geq T_0$, where
\begin{equation}\label{Eqn:BoundsOnT0}
    T_0 =
    \frac{(4\pi)^{6/5}R_1^{2/5}R_2^2}{\varepsilon^{4/5}C_0^{6/5}\nu} ~.
\end{equation}
\end{proposition}

The RHS has units of time. If there is a nonzero force present, 
then $R_1 = R_1(T)$ and $R_2=R_2(T)$, so that the expressions 
\eqref{Eqn:BoundsOnKappa1}\eqref{Eqn:BoundsOnKappa2} for
$\overline{\kappa}_1=\overline{\kappa}_1(T)$ and 
$\overline{\kappa}_2 = \overline{\kappa}_2(T)$ depend 
on time. It nonetheless could happen that 
\begin{equation}
   \limsup_{T\to +\infty} \overline{\kappa}_2(T) 
  < \liminf_{T\to +\infty}\overline{\kappa}_1(T) ~,
\end{equation}
then again there is a maximum time $T_0$ for the existence of spectral 
behavior of solutions. 

The above three estimates give lower and
upper bounds on the inertial range, and an upper bound of the time of 
validity of a spectral description of a solution to 
\eqref{Eqn:NavierStokes}, if indeed it behaved exactly like the 
Kolmogorov power spectrum profile over its inertial range. 

As discussed in the introduction, when additional physical assumptions
are made as to the behavior of a solution of \eqref{Eqn:NavierStokes},
then further information is available about the energy transfer rate 
$\varepsilon$. The classical hypotheses are stated in \cite{Obukhov1941} 
among other places, describing the character of flows in a regime of fully 
developed turbulence. Specifically, they are that (1) the flow is in a
(statistically) steady state of energy transfer from the inertial
range $|k| \leq \kappa_\nu$ to the dissipative range $|k| > \kappa_\nu$;
(2) the support of the spectrum $E(\kappa,t)$ lies essentially in 
inertial range; and (3) a certain scale invariant form is assumed 
for the transport of energy $T(\kappa)$ at wavenumber scale $\kappa$
which assumes a form of homogeneity of the flow. 
Under hypotheses (1)(2) and (3), one concludes 
that the energy dissipation rate $\varepsilon_1$ in 
\eqref{Eqn:DissipationRatePerVolume} is equal to the energy 
transfer rate $\varepsilon$, and that the upper end of the inertial
range $\kappa_2 = \kappa_\nu = 2\pi(\varepsilon/\nu^3)^{1/4}$. 
Alternatively, one can simply work 
under the hypothesis that the energy spectral function depends only 
upon the two quantities $\kappa$ and the energy dissipation rate 
$\varepsilon_1$. Assuming that the force is stationary, and that
$\varepsilon = \varepsilon_1$, the conclusion of 
\cite{DoeringFoias02} is that
\[
    \varepsilon \leq c_1\nu 
         \frac{\langle\|u(\cdot )\|^2_{L^2} \rangle}{V^{5/3}} 
  + c_2 \frac{\langle\|u(\cdot )\|^2_{L^2} \rangle^{3/2}}{V^{11/6}}~,
\] 
where $\langle \cdot \rangle$ denotes either time or ensemble
averaging, and in any case $\|u(\cdot )\|_{L^2} \leq R$. This is 
to be  compared with the general estimate 
\eqref{Eqn:GraphIntersection}, and it behaves better for small
$\nu$. On the other hand, the estimates in the present paper hold 
for all weak solutions of \eqref{Eqn:NavierStokes}, and they give
an upper bound on the energy transfer rate $\varepsilon$, essentially
independently of the Obukhov hypotheses. Furthermore, the upper and
lower bounds on the inertial range are conclusions of the analysis
rather than assumptions of the theory. In particular this work gives 
a lower bound on the inertial range, which is the first such known, 
either with mathematically rigorous arguments or under physical 
hypotheses, at least to the authors. 

\subsection{Limits on spectral behavior}
The endpoints $\overline{\kappa}_1, \overline{\kappa}_2$ of the 
interval which bounds the inertial range, and the 
temporal upper bound $T_0$ are given in terms of the idealized 
Kolmogorov spectral function $E_K(\kappa)$, rather than one
given by an actual solution of the Navier -- Stokes equations. 
In order to have a relevance to actual solutions, one must quantify
the meaning of {\it spectral behavior} of a solution. This can have a
number of interpretations, several of which we have mentioned in 
Section~\ref{Section:KSpectrum}. It could be that we define an
individual solution to have spectral behavior if its energy spectral 
function $E(\kappa,t)$ is sufficiently close to the idealized Kolmogorov 
spectral function $E_K(\kappa)$, uniformly over a time period $[0,T]$. 
Since the solution $u(\cdot,t) \in L^2$ and is indeed in $\dot H^1$
for almost all times $t$, while $E_K$ is neither ({\it i.e.} neither
$E_K(\kappa)$ nor $\kappa^2E_K(\kappa) \in L^1({\mathbb R}^+_\kappa)$), 
this already implies that the inertial range must be finite. 

\begin{definition}
A solution $(u,p)$ to \eqref{Eqn:NavierStokes} is said to have 
the spectral behavior of $E_K(\kappa)$, uniformly over the range 
$[{\kappa}_1, {\kappa}_2]$ and for the time 
interval $[0,\overline{T}]$ if its energy spectral function
$E(\kappa,t)$ satisfies 
\begin{equation}\label{Eqn:UniformSpectralBehavior}
   \sup_{\kappa\in [{\kappa}_1,{\kappa}_2], t\in [0,\overline{T}]}
     (1+\kappa^{5/3}) | E(\kappa,t) - E_K(\kappa)| < C_1 <\!\!< C_0\varepsilon^{2/3}~.
\end{equation}
\end{definition}

An alternate version of this specification would be to replace the
criterion \eqref{Eqn:UniformSpectralBehavior} with a weaker one, 
for instance asking that a Sobolev space norm be controlled, 
which for $D = {\mathbb R}^3$ could be the statement that
\begin{equation}\label{Eqn:SobolevSpaceCriterion}
     \sup_{t\in [0,\overline{T}]} \int_{\kappa_1}^{\kappa_2}
     (1+\kappa^{5/3}) | E(\kappa,t) - E_K(\kappa)| \, d\kappa 
  < C_1 <\!\!< C_0\varepsilon^{2/3}~.
\end{equation}
Or else one could specify a criterion which respected the metric 
of a Besov space. For example, one could use the Fourier decomposition 
$\Delta_j = \{ k \ : \ 2^{j-1/2} < |k| \leq 2^{j+1/2} \}$, and 
ask that over a time period $[0,\overline{T}]$ a solution satisfy
\begin{equation}\label{Eqn:BesovSpaceCriterion}
   \Bigl| \int_{\Delta_j} |\hat{u}(k,t)|^2 \, dk 
   - C_0\varepsilon^{2/3}3\sinh((\ln 2)/3) 2^{-2j/3} \Bigr|
   < C_1 2^{-5j/3} \ll C_0\varepsilon^{2/3} 2^{-5j/3}    
\end{equation}
for all $j_1 \leq j \leq j_2$, where $j_1,j_2$ are such that 
$j_1 < \log_2 {\kappa}_1$, and $\log_2 {\kappa}_2 < j_2$. 
In any of these cases, Theorems~\ref{Thm:four} and \ref{Thm:five} imply 
bounds on the inertial range given by the interval 
$[\overline{\kappa}_1, \overline{\kappa}_2]$ .

\begin{theorem}\label{Thm:six}
Suppose that an individual solution $(u(x,t),p(x,t))$ is such that 
$u_0(x) \in A_{R_1}\cap B_R(0)$, where $R$ and $R_1$ satisfy 
\eqref{Eq:EstimateOnA} and if a force is present, it satisfies 
 \eqref{Eq:EstimateOnAandF}. If $u(x,t)$ exhibits the spectral 
behavior of $E_K$ uniformly over the range 
$[{\kappa}_1, {\kappa}_2]\times[0,\overline{T}]$, 
then 
\begin{equation}\label{Eqn:BoundsOnSpectralBehavior}
   \overline{\kappa}_1 \leq {\kappa}_1 ~, \quad 
   {\kappa}_2 \leq \overline{\kappa}_2 ~, \quad
\end{equation}
and if $f=0$ then
\begin{equation}
      \overline{T} \leq T_0 ~, 
\end{equation}
with possibly different constants $C_0$ in 
\eqref{Eqn:BoundsOnKappa1}\eqref{Eqn:BoundsOnKappa2}\eqref{Eqn:BoundsOnT0}. 
\end{theorem}

Thus the spectral behavior of solutions whose initial data $u_0(x)$
lie in one of the sets $A_{R_1}\cap B_R(0)$ is limited by the bounds
given in \eqref{Eqn:BoundsOnSpectralBehavior}. The proof will show 
that the same constraints hold for solutions which exhibit spectral 
behavior over a given nonzero proportion of the measure of the time 
interval $[0,\overline{T}]$.

However it could be argued that the behavior of an individual solution
is less important, and that spectral behavior is a property of a
statistical ensemble of solutions. Members of this ensemble 
should have their spectral behavior considered in terms of 
the ensemble average, rather than individually as above. 
Theorems~\ref{Thm:four} and \ref{Thm:five} are relevant to
this situation as well. Suppose there were a probability measure
${\sf P}$ defined on a statistical ensemble 
$\Omega\subset L^2(D)\cap\{ u \ : \ \nabla \cdot u = 0\}$ 
which is invariant under the solution map of the Navier -- Stokes
equations, however this has been chosen to be defined, with 
force $f$ (also possibly stationary, taken from a family of 
realizations which have their own statistics). Using the 
standard notation, define the ensemble average of a functional $F(u)$ 
defined and ${\sf P}$-measurable on $\Omega$ by $\langle F(u) \rangle$. 
Without loss of generality we can take ${\sf P}$ to be ergodic
with respect to the Navier -- Stokes solution map. 

The ergodicity of the invariant measure ${\sf P}$ tells us two
things. The first is that space averages are {\it a.e.} time 
averages, so that 
\begin{equation}
   \langle E(\kappa,\cdot)\rangle = \lim_{T\to\infty} \frac{1}{T}
     \int_0^T E(\kappa,t) \, dt \leq \frac{4\pi}{\nu\kappa^2 }
       \lim_{T\to\infty} \frac{R_2^2(T)}{T}
\end{equation}
for ${\sf P}$-{\it a.e.} initial data $u_0$. The second thing is 
that whenever $R, R_1$ satisfy \eqref{Eq:EstimateOnAandF}
then ${\sf P}(A_{R_1}\cap B_R(0))$ is either zero or one, as 
$A_{R_1}\cap B_R(0)$ is an invariant set. 

\begin{definition}
A statistical ensemble $(\Omega,{\sf P})$ is said to exhibit 
the spectral behavior of $E_K(\kappa)$ on average over the range 
$[{\kappa}_1, {\kappa}_2]$ when the ensemble 
average of its energy spectral function, 
\begin{equation}
   \langle E(\kappa,t)\rangle := \int_{|k| = \kappa} 
    \langle |\hat{u}(k,t)|^2\rangle \, dS(k)
\end{equation}
satisfies the estimate
\begin{equation}\label{Eqn:AverageSpectralBehavior}
   \sup_{\kappa\in [{\kappa}_1,{\kappa}_2],t\in [0,\overline{T}]}
     (1+\kappa^{5/3}) | \langle E(\kappa,t)\rangle - E_K(\kappa)| < C_1 
   <\!\!< C_0\varepsilon^{2/3}
\end{equation}
over the range $[{\kappa}_1, {\kappa}_2]$.
\end{definition} 

Let us suppose that the force $f$ satisfies 
\begin{equation}
   |\hat{f}(k,t)| \leq \nu |k| R_1 \qquad \hbox{\rm and} \quad 
   \bigl(\int_0^t|\hat{f}(k,s)|^2 \, ds \bigr)^{1/2} \leq F_\infty(t) ~,
\end{equation}
as in \eqref{Eq:EstimateOnAandF}, and we are to examine the spectral
behavior of the statistical ensemble of solutions $\{ u(\cdot)\}$.

\begin{theorem}\label{Thm:seven}
Suppose that the ensemble $(\Omega, {\sf P})$ has the spectral
behavior of $E_K(\cdot)$ over the range 
$[{\kappa}_1, {\kappa}_2]$. Then either 
\begin{equation}
   {\sf P}(A_{R_1}\cap B_R(0)) = 0
\end{equation}
for all $R, R_1$, or else 
\begin{equation}\label{Eqn:BoundsOnEnsembleSpectralBehavior}
   \overline{\kappa}_1 \leq {\kappa}_1 ~, \qquad 
   {\kappa}_2 \leq \overline{\kappa}_2 ~, 
\end{equation}
with a possibly different constant $C_0$ in 
\eqref{Eqn:BoundsOnKappa1}\eqref{Eqn:BoundsOnKappa2}\eqref{Eqn:BoundsOnT0}. 
\end{theorem}

\begin{proof}[Proof of Theorems \ref{Thm:six} and \ref{Thm:seven}]
We will give the argument in the case of Euclidian space 
$D = {\mathbb R}^3$, the torus case is similar. 
Start with the proof of Theorem \ref{Thm:six} with the criterion 
of \eqref{Eqn:UniformSpectralBehavior}, and suppose that 
${\kappa}_1 < \overline{\kappa}_1$. Using the estimate of 
Theorem \ref{Thm:four} and the form 
\eqref{Eqn:IdealizedKolmogorovEnergySpectrum}
of $E_K({\kappa}_1)$ we have 
\[
   C_0\varepsilon^{2/3}{\kappa}_1^{-5/3} - 4\pi R_1^2 
  \leq | E_K({\kappa}_1) - E({\kappa}_1) |
   \leq o(1) C_0\varepsilon^{2/3} ~.
\]
Because of the identity \eqref{Eqn:DefinitionKappa1}, this implies
\[
   C_0\varepsilon^{2/3}({\kappa}_1^{-5/3} - \overline{\kappa}_1^{-5/3}) 
   \leq o(1)C_0\varepsilon^{2/3} ~.
\]
A lower bound for the LHS is given by
\[
   {\kappa}_1^{-5/3} - \overline{\kappa}_1^{-5/3} 
  = \int_{{\kappa}_1}^{\overline{\kappa}_1} \frac{5}{3}\kappa^{-8/3} \,
    d\kappa
  \geq (\overline{\kappa}_1 - {\kappa}_1)\frac{5}{3}\overline{\kappa}_1^{-8/3} ~.
\]
Therefore 
\[
   0 \leq (\overline{\kappa}_1 - {\kappa}_1) 
   \leq o(1)\frac{3}{5}\overline{\kappa}_1^{8/3}~,
\]
that is, ${\kappa}_1$ is a bounded distance from $\overline{\kappa}_1$. 
Furthermore the defining relation \eqref{Eqn:DefinitionKappa1} for 
the left endpoint $\overline{\kappa_1}$ of the bounds on the inertial 
range can be rewritten
\[
  4\pi R_1^2 = C_0\varepsilon^{2/3}\overline{\kappa}_1^{-5/3} = 
   C_0\varepsilon^{2/3}{\kappa}_1^{-5/3}
   \Bigl(\frac{{\kappa}_1}{\overline{\kappa_1}}\Bigr)^{5/3} ~, 
\]
and since $(1- o(1)(3/5)\overline{\kappa}_1^{5/3}) 
\leq {\kappa}_1 / \overline{\kappa}_1 \leq 1$, then
\eqref{Eqn:DefinitionKappa1} continues to hold for $\kappa_1$, 
with only a small change in the constant $C_0$. 

Now suppose that $\overline{\kappa}_2 \leq {\kappa}_2$. The criterion
\eqref{Eqn:UniformSpectralBehavior} implies that
\[
   \frac{1}{\overline{T}}\int_0^{\overline{T}} 
      \kappa^{5/3} |E_K(\kappa) - E(\kappa,t)| \, dt
   \leq o(1) C_0\varepsilon^{2/3} ~.
\]
Therefore using \eqref{Eqn:EnergySpectrumDecay} and 
\eqref{Eqn:IdealizedKolmogorovEnergySpectrum}, we have
\[
   C_0\varepsilon^{2/3} - \frac{4\pi}{\nu}\frac{R_2^2}{T}\kappa^{-1/3} 
   \leq o(1)C_0\varepsilon^{2/3} ~.
\]
for $\kappa \in [{\kappa}_1,{\kappa}_2]$ and $T \leq \overline{T}$. 
This applies in particular to $\kappa = \overline{\kappa}_2$, 
therefore 
\[
   C_0\varepsilon^{2/3} (1 - o(1)) 
   \leq \frac{4\pi}{\nu}\frac{R_2^2}{T}\frac{1}{{\kappa}_2^{1/3}} ~.
\]
Hence
\[
   {\kappa}_2^{1/3} 
   \leq \frac{4\pi}{\nu}\frac{R_2^2(1 + o(1))}{C_0\varepsilon^{2/3}T}
   = (1 + o(1))\overline{\kappa}_2^{1/3} ~,
\]
where we have used \eqref{Eqn:BoundsOnKappa2}. Therefore 
\[
   (1 - o(1)) \leq \frac{\overline{\kappa}_2}{{\kappa}_2} \leq 1 ~, 
\]
and \eqref{Eqn:BoundsOnKappa2} holds for ${\kappa}_2$ with 
only a change of the overall constant $C_0$. 

Similar considerations give the analog result to Theorem~\ref{Thm:six}
if we accept the Sobolev or Besov criteria for spectral behavior. For
instance, suppose it is considered that the estimate 
\eqref{Eqn:SobolevSpaceCriterion} is the indicator of 
spectral behavior. If ${\kappa}_1 < \overline{\kappa}_1$ then
\[
   \int_{{\kappa}_1}^{\overline{\kappa}_1} 
      ( C_0\varepsilon^{2/3} \kappa^{-5/3} - 4\pi R_1^2 ) \, d\kappa
   \leq \int_{{\kappa}_1}^{\overline{\kappa}_1}
     | E_K(\kappa) - E(\kappa) | \, d\kappa 
   \leq o(1) C_0\varepsilon^{2/3} ~. 
\]
Using \eqref{Eqn:DefinitionKappa1}, this implies that 
\[
    C_0\varepsilon^{2/3} \int_{{\kappa}_1}^{\overline{\kappa}_1}
    (\kappa^{-5/3} - \overline{\kappa}_1^{-5/3}) \, d\kappa 
   \leq  o(1) C_0\varepsilon^{2/3} ~,
\]
which in turn implies (by convexity) that
\[
   C_0\varepsilon^{2/3} \frac{(\overline{\kappa}_1 - {\kappa}_1)^2}{2}
   \frac{5}{3} \overline{\kappa}_1^{-8/3} \leq o(1) C_0\varepsilon^{2/3} ~. 
\]
This controls $\overline{\kappa}_1 - {\kappa}_1$ and also their ratio. 
Suppose that $\overline{\kappa}_2 < {\kappa}_2$. Then the criterion
\eqref{Eqn:SobolevSpaceCriterion} implies that 
\[
  \int_{\overline{\kappa}_2}^{{\kappa}_2} \kappa^{5/3} 
    (C_0\varepsilon^{2/3}\kappa^{-5/3} 
   - \frac{4\pi R_2^2}{\nu\overline{T}}\kappa^{-2}) \, d\kappa
  \leq \int_{\overline{\kappa}_2}^{{\kappa}_2} \kappa^{5/3} 
    | E_K(\kappa) - E(\kappa) | \, d\kappa 
   \leq o(1) C_0\varepsilon^{2/3} ~.
\]
Therefore, using \eqref{Eqn:BoundsOnKappa2} 
\[
   ( {\kappa}_2 - \overline{\kappa}_2) - \frac{3}{2}\overline{\kappa}_2^{1/3} 
     ({\kappa}_2^{2/3} - \overline{\kappa}_2^{2/3}) \leq o(1) ~. 
\]
Define $a(\kappa) := \kappa - \frac{3}{2}\overline{\kappa}_2^{1/3}\kappa^{2/3}$,
which is increasing and convex for $\kappa \geq
\overline{\kappa}_2$. The last estimate states that 
$a({\kappa}_2) - a(\overline{\kappa}_2) \leq o(1)$, which 
controls the quantity ${\kappa}_2 - \overline{\kappa}_2$.  

The proof of the analogous statements of Theorem~\ref{Thm:seven}
are similar, except that the upper bounds on ${\kappa}_2$
are easier as the ensemble average already subsumes the time average
due to the ergodicity hypothesis. \qquad 
\end{proof}

\section{Conclusions}

The global estimates given in theorems \ref{Thm:one}, \ref{Thm:two} 
and \ref{Thm:TimeAverageEnergySpectrum} for the domain
$D = {\mathbb T}^3$, and theorems \ref{Thm:ContinuumL-inftyEstimate} 
and \ref{Thm:ContinuumL-2TEstimate} in the case $D = {\mathbb R}^3$,
provide control in $L^\infty$ of the Fourier transform of weak
solutions of the Navier -- Stokes equations. These are in terms of
constants $R$, $R_1$ and $R_2$ which depend only upon the initial data
and the inhomogeneous forces. These results in turn give estimates of 
the energy spectral function, which show that $E(\kappa,t)$ is bounded
from above, and its time averages are bounded above by 
${\mathcal O}(1/\kappa^2)$. These upper bounds constrain the ability for 
a weak solution to exhibit spectral behavior in the manner of the 
idealized Kolmogorov spectral function 
$E_K(\kappa) = C_0\varepsilon^{2/3}\kappa^{-5/3}$. The constraints
extend to the case of a statistical ensemble forces and solutions, 
applying to the ensemble averages $\langle E(\kappa,t)\rangle$ of 
the energy spectral function. We remark that the estimates, and the
subsequent constraints on spectral behavior, are valid for weak
solutions of the Navier -- Stokes equations, and our considerations
are separate from the physical assumptions of Obukhov on flows
exhibiting fully developed turbulence, or the question of possible 
formation of singularities. 

It is natural to compare the above constraints with the physical
quantities describing spectral behavior and the inertial range that
come from the Kolmogorov -- Obukhov theory of turbulence. The first of
these is the Kolmogorov length scale 
$\eta_\nu = (\nu^3/\varepsilon)^{1/4}$, or rather its associated 
wavenumber $\kappa_\nu = 2\pi / \eta_\nu$. On physical grounds, 
dissipation is expected to dominate the behavior of $E(\kappa,t)$
for $\kappa > \kappa_\nu$. Comparing $\kappa_\nu$ to our upper bounds
on the inertial range, we find that
\[
   \kappa_\nu = 2\pi \left( \frac{\varepsilon}{\nu^3} \right)^{1/4} 
  \leq \left( \frac{4\pi}{C_0\nu}\frac{R_2^2(T)}{T}\right)^3 
   \frac{1}{\varepsilon^2} = \overline{\kappa}_2 ~
\]
for sufficiently small $\varepsilon$ and $\nu$. Indeed, with
everything else fixed, $\kappa_\nu$ is decreasing in $\varepsilon$ 
while $\overline{\kappa}_2$ is increasing, and furthermore while both 
$\kappa_\nu$ and  $\overline{\kappa}_2$ are increasing as $\nu \to 0$,
however $\kappa_\nu \ll \overline{\kappa}_2$. It seems 
clear that $\overline{\kappa}_2$ is an absolute, but not necessarily a very 
sharp, estimate of the upper limit of the inertial range and the 
start of the dissipative regime for solutions that is 
expected on physical grounds. 

As described in section~\ref{SubSection3.3}, if one assumes 
a number of physical hypotheses, as Obukhov does, on the form
of the energy transfer rate, from which one deduces that 
$\kappa_\nu$ gives an upper bound on the spectral regime and 
that $\varepsilon = \varepsilon_1$ is given by the energy 
dissipation rate, then there are better upper bounds available for 
$\varepsilon$ than our result \eqref{Eqn:GraphIntersection}. This
assumption however is based on physical assumptions on 
the character of solutions of the Navier -- Stokes equations in 
a statistically stationary regime. 

The Taylor length scale 
$\kappa_\lambda = 2\pi (\varepsilon V /\nu R^2)^{1/2}$ is another
indicator of the lower limit of the dissipative regime, one which 
incidentally is independent of the form of the Kolmogorov 
idealized energy spectral function $E_K$. The quantity $R^2/V$ 
is a bound on the energy per unit volume of the solution. 
We again see that $\overline{\kappa}_2$ 
is an overly pessimistic upper bound for $\kappa_\lambda$
for small $\varepsilon$ and $\nu$, since in such a case
\[
   \kappa_\lambda = 2\pi \left( \frac{\varepsilon V}{\nu R^2} \right)^{1/2}
     \leq \left( \frac{4\pi}{C_0\nu}\frac{R_2^2(T)}{T}\right)^3 
          \frac{1}{\varepsilon^2} = \overline{\kappa}_2 ~.
\]
In both of these comparisons the quantities $R_2^2(T)/T$ are to be
replaced by $\overline{R}_2^2$ in the case of a statistically 
stationary ensemble of solutions. 

In a flow regime of fully developed turbulence, it is generally 
expected that $\kappa_\lambda < \kappa_\nu$, an inequality which 
is worthwhile to discuss. Calculate 
\[
   \frac{\kappa_\lambda}{\kappa_\nu} 
   = \frac{\sqrt{V}}{R} \left(\frac{\varepsilon}{\nu}\right)^{1/2}
     \left( \frac{\nu^3}{\varepsilon}\right)^{1/4} 
   = \frac{\sqrt{V}}{R} (\varepsilon \nu)^{1/4} ~
\] 
where $(\varepsilon \nu)^{1/4} = u_\nu$ is the Kolmogorov velocity
scale. For solutions that we consider, \eqref{Eq:EstimateOnA} holds, 
so that in particular 
\[
   \frac{\sqrt{V}}{R} (\varepsilon \nu)^{1/4} 
  \geq \frac{R}{\nu R_1}(\varepsilon \nu)^{1/4}
  = \frac{R}{R_1}\frac{\kappa_\nu}{2\pi} ~. 
\]   
The implication is that 
\[
   \kappa_\lambda > \frac{R}{2\pi R_1}\kappa_\nu^2 ~,
\]
which indicates, for fixed data $R, R_1$, that $\kappa_\lambda$
cannot be too much smaller than $\kappa_\nu$, and is very possibly
much larger. An inequality in the
other sense does not seem to arise from this or similar
considerations. We do find an upper bound for the Kolmogorov 
velocity scale
\[
   u_\nu = (\varepsilon\nu)^{1/4} 
   = \frac{R}{\sqrt{V}} \frac{\kappa_\lambda}{\kappa_\nu}
   \leq  \frac{1}{\nu^{1/16}} \left( \frac{4\pi}{C_0} \right)^{3/8}
        \left(\left(\frac{R_2}{\sqrt{T}}\right)^5 R_1 \right)^{1/8} ~.
\]

Considering the case $f=0$, which is as in the original papers of 
Kolmogorov~\cite{K41c}, the constraint of Theorem~\ref{Thm:six} 
is that $\overline{T} \leq T_0$, with the latter given by the
expression in \eqref{Eqn:BoundsOnT0}. This is to be compared with the
Kolmogorov timescale $\tau_\nu = (\nu/\varepsilon)^{1/2}$. It is
clear, for $R_1$ and $R_2$ fixed constants, that 
\[
   \frac{\tau_\nu}{T_0} = \varepsilon^{3/10} \nu^{3/2}
     \left(\frac{C_0^{6/5}}{(4\pi)^{6/5} R_1^{2/5}R_2^2} \right)
\]
which is of course small for small $\varepsilon, \nu$. This is again
as it should be, allowing large multiples of the eddy turnover time
before one runs into the upper allowed limit for the persistence of
spectral behavior of solutions. 

It is also natural, given the bounds $\overline{\kappa}_1$, $\overline{\kappa}_2$ on the 
inertial range, to introduce the dimensionless parameter 
\begin{equation}
   r_\nu := \frac{\overline{\kappa}_2}{\overline{\kappa}_1} 
   = \frac{1}{\varepsilon^{12/5}\nu^3} 
     \left(\frac{4\pi}{C_0} \right)^{18/5}      
     \left(\frac{R_1^{2/5} R_2^2(T)}{T} \right)^3 ~,
\end{equation}
which governs the extent of the possibility of spectral behavior of
solutions. It is somewhat similar to a Reynold's number; 
when $r_\nu < 1$ then solutions satisfying 
\eqref{Eq:EstimateOnA}\eqref{Eq:EstimateOnAandF} (respectively,
\eqref{Eqn:ContinuumBoundednessRelation}) are disallowed from 
exhibiting spectral behavior. For $r_\nu > 1$ an inertial range
is permitted, although it is not guaranteed by the analysis of this
paper. The larger $r_\nu$ the larger the permitted inertial range,
although again it is not the case that the actual interval of $\kappa$ 
over which solutions exhibit spectral behavior will necessarily extend 
through a significant proportion of the interval 
$\overline{\kappa}_1, \overline{\kappa}_2$.
In the situation of a statistical ensemble of forces and solutions,
the form of $r_\nu$ is somewhat more compelling,
\begin{equation}
      r_\nu := \frac{\overline{\kappa}_2}{\overline{\kappa}_1} 
   = \frac{1}{\varepsilon^{12/5}\nu^3} 
     \left(\frac{4\pi}{C_0} \right)^{18/5}      
     \left( R_1^{2/5}\overline{R}_2^2 \right)^3 ~.
\end{equation}
This quantity is a stand-in for the ratio of the integral scale 
to the Kolmogorov scale, which is itself often used as an indicator
of the Reynolds number of a flow. 

This paper does not address the corrections to the Kolmogorov -- Obukhov
theory of Navier -- Stokes flows in a turbulent 
regime, along the lines proposed in Kolmogorov (1962) \cite{K62}. 
This is focused on the deviations from Gaussian nature of the 
moments of the structure function for such flows, and it has 
been a very active area of research over the past decades. 
We will reserve our own thoughts on this matter for a future 
publication.

          
  
  

  

\ack   
WC would like to thank N. Kevlahan and B. Protas for many edifying
conversations on Navier -- Stokes turbulence over the course of the 
past several years. Our research is supported in part by the Canada 
Research Chairs Program and by NSERC Discovery Grant \#238452-06.

  
\frenchspacing  
\bibliographystyle{plain}

  
  
                                  

\end{document}